\documentclass[letterpaper, 11pt]{article}
\usepackage{appendix}
 
\newcommand{\comment}[1]{}
 
\usepackage{graphicx}

\usepackage{subfigure}
\usepackage{amsmath}
\usepackage{amssymb}
\usepackage{mdwlist}

\usepackage{algorithm}

\usepackage{hyperref}

\usepackage{xspace}
\usepackage{setspace}

\newcommand{\TRUE}{{\tt TRUE}\xspace}
\newcommand{\FALSE}{{\tt FALSE}\xspace}

\newcommand{\sv}{V}

\newenvironment{proof}{\noindent {\bf Proof:}~}{\hspace*{\fill}\(\Box\)}

\newtheorem{theorem}{Theorem}

\newtheorem{claim}{Claim}

\newtheorem{definition}{Definition}

\newtheorem{lemma}{Lemma}

\def\noflash#1{\setbox0=\hbox{#1}\hbox to 1\wd0{\hfill}}

\newcommand{\vectorv}{{\normalfont\textbf{v}}}
\newcommand{\matrixm}{\textbf{M}}

\newcommand{\D}{{\textbf d_H}}

\newcommand{\R}{{\mbox{\it Verified\,}}}

\newcommand{\HH}{{\mathcal H}}
\newcommand{\VV}{{\mathcal V}}
\newcommand{\II}{{\mathcal I}}

\newcommand{\SVRecv}{{\tt SVRecv}}
\newcommand{\RBRecv}{{\tt RBRecv}}
\newcommand{\RBSend}{{\tt RBSend}}
\newcommand{\Verify}{{\tt Verify\,}}
\newcommand{\Proceed}{{\tt Proceed\,}}
\newcommand{\Add}{{\tt Add\,}}

\newcommand{\bfA}{{\bf A}}

\newcommand{\bfM}{{\bf M}}
\newcommand{\bfv}{v}

\begin{document}
\title{Byzantine Convex Consensus: An Optimal Algorithm\footnote{\normalsize This research is supported in part by National Science Foundation award CNS 1059540. Any opinions, findings, and conclusions or recommendations expressed here are those of the authors and do not necessarily reflect the views of the funding agencies or the U.S. government.}}

\author{Lewis Tseng$^{1,3}$, and Nitin Vaidya$^{2,3}$\\~\\
 \normalsize $^1$ Department of Computer Science,\\
 \normalsize $^2$ Department of Electrical and Computer Engineering, 
 and\\ \normalsize $^3$ Coordinated Science Laboratory\\ \normalsize University of Illinois at Urbana-Champaign\\~\\ \normalsize Email: \{ltseng3, 
nhv\}@illinois.edu \\~\\ \normalsize Technical Report} 


\date{July 9, 2013}
\maketitle


\comment{
\noindent
{\bf
\begin{itemize}
\item This is a regular paper submission.
\item Lewis Tseng is a full time Ph.D student at the University of Illinois\\at Urbana-Champaign.
\item This paper is eligible for the best student paper award.
\end{itemize}
}
}

\begin{abstract}
{\normalsize

Much of the past work on asynchronous approximate Byzantine consensus has
assumed {\em scalar} inputs at the nodes \cite{AA_Dolev_1986, AA_nancy}.
Recent work has yielded approximate Byzantine consensus algorithms
for the case when the input at each node is a $d$-dimensional vector,
and the nodes must reach consensus on
a vector in the convex hull
of the input vectors at the fault-free nodes
\cite{herlihy_multi-dimension_AA,Vaidya_BVC}.
The $d$-dimensional vectors can be equivalently viewed as {\em points}
in the $d$-dimensional Euclidean space.
Thus, the algorithms in \cite{herlihy_multi-dimension_AA, Vaidya_BVC}
require the fault-free nodes to decide on a point
in the $d$-dimensional space.

~

In our recent work \cite{Tseng_BCC}, we proposed a generalization of the consensus problem, namely {\em Byzantine convex consensus} (BCC), which allows the decision to be
a {\em convex polytope} in the $d$-dimensional space, such that the decided
polytope is within the convex hull of the input vectors at the fault-free nodes.
We also presented an asynchronous approximate BCC algorithm. 

~

In this paper, we propose a new BCC algorithm with optimal fault-tolerance that also agrees on a convex polytope that is as {\em large} as possible under adversarial conditions. Our prior work \cite{Tseng_BCC} does not guarantee the optimality of the output polytope.

}
\end{abstract}

\thispagestyle{empty}
\newpage
\setcounter{page}{1}

\section{Introduction}
\label{s_intro}

Much of the past work on asynchronous approximate Byzantine consensus has
assumed {\em scalar} inputs at the nodes \cite{AA_Dolev_1986, AA_nancy}.
Recent work has yielded approximate Byzantine consensus algorithms
for the case when the input at each node is a $d$-dimensional vector,
and the nodes must reach consensus on
a vector in the convex hull
of the input vectors at the fault-free nodes
\cite{herlihy_multi-dimension_AA,Vaidya_BVC}.
The $d$-dimensional vectors can be equivalently viewed as {\em points}
in the $d$-dimensional Euclidean space.
Thus, the algorithms in \cite{herlihy_multi-dimension_AA, Vaidya_BVC}
require the fault-free nodes to decide on a point
in the $d$-dimensional space.
In our recent work \cite{Tseng_BCC}, we considered a generalized problem, namely {\em Byzantine convex consensus} (BCC), which allows the decision to be a {\em convex polytope} in the $d$-dimensional space, such that the decided polytope is within the convex hull of the input vectors at the fault-free nodes. In this paper, we propose an asynchronous BCC algorithm with optimal fault-tolerance that reaches consensus on the convex polytope with an {\em optimal} output polytope (as defined later).
This is an improvement over our previous algorithm in \cite{Tseng_BCC} that does not guarantee optimality of the output polytope.

The system under consideration is
an {\em asynchronous} system consisting of $n$ nodes,
of which at most $f$ may be Byzantine faulty. The Byzantine faulty nodes may behave
in an arbitrary fashion,
and may collude with each other. Each node $i$ has a $d$-dimensional vector of reals as its {\em input} $x_i$.
All nodes can communicate with each other directly on reliable and FIFO (first-in first-out)
channels.  Thus, the underlying communication graph can be modeled as a {\em complete graph}, 
with the set of nodes being $V=\{1,2,\cdots, n\}$.
The impossibility of {\em exact} consensus in asynchronous systems \cite{FLP_one_crash}
applies to BCC as well. Therefore, we consider 
the {\em Approximate BCC} problem with the following requirements:

\begin{itemize}

\item \textbf{Validity}: The {\em output} (or {\em decision}) at each fault-free node must be
a convex polytope in the convex hull of the $d$-dimensional input vectors at the fault-free nodes.
(In a degenerate case, the output polytope may simply be a single {\em point}.)

\item \textbf{$\epsilon$-Agreement}: For any $\epsilon > 0$, the {\em Hausdorff distance} (defined below) between the output polytopes at any two fault-free nodes must be at most $\epsilon$.

\item {\bf Termination:} Each fault-free node must terminate within a finite amount of time.


\end{itemize}
The motivation behind reaching consensus on a convex polytope is that
a solution to BCC is expected to also facilitate solutions to a large
range of consensus problems (e.g., Byzantine vector consensus \cite{herlihy_multi-dimension_AA,Vaidya_BVC}, or convex function optimization
over a convex hull of the inputs at fault-free nodes).
Future work will explore these potential applications.


\begin{definition}
\label{def:dist}
For two convex polytopes $h_1, h_2$, the {\em Hausdorff distance} is defined as \cite{Hausdorff}
\[
\D(h_1, h_2) ~~=~~ \max~~ \{~~ \max_{p_1 \in h_1}~\min_{p_2 \in h_2} d(p_1, p_2),~~~~ \max_{p_2 \in h_2}~\min_{p_1 \in h_1} d(p_1, p_2)~~\}
\]
\end{definition}
where $d(p,q)$ is the Euclidean distance between points $p$ and $q$. 

\paragraph{Optimality of the Output Polytope:}
The BCC algorithm proposed in this paper allows the nodes to agree on an output polytope that is ``optimal'' in the sense defined below.

\begin{definition}
\label{def:optSize}
Let $A$ be an algorithm that solves Byzantine convex consensus (BCC). Algorithm $A$ is said to reach consensus on an {\em optimal} convex polytope if for any BCC algorithm $B$, there exist a behavior of the faulty nodes and a message delay pattern such that, at each fault-free node, the output polytope obtained using algorithm $B$ is contained in the output polytope obtained using algorithm $A$.
\end{definition}

We show that the BCC algorithm proposed here allows the nodes to agree on a polytope that is guaranteed to contain a polytope that is named $I_Z$ in later analysis. $I_Z$ is a function of the inputs at some of the fault-free nodes. We show that, for any correct BCC algorithm, there exists an execution in which the fault-free nodes must agree on a polytope that is equal to or contained in $I_Z$. Thus, as per Definition \ref{def:optSize}, the output polytope chosen by our algorithm is optimal.


\paragraph{Lower Bound on $n$:}
As noted above,
\cite{herlihy_multi-dimension_AA, Vaidya_BVC} consider the problem of
reaching approximate Byzantine consensus on a
vector (or a point) in the convex hull of the $d$-dimensional
input vectors at the fault-free nodes, and show that
$n \geq (d+2)f+1$ is necessary.
\cite{herlihy_colorless_async} generalizes
the same lower bound to colorless tasks.
The lower bound proof in
\cite{herlihy_multi-dimension_AA, Vaidya_BVC} also implies that
$n \geq (d+2)f+1$ is necessary to ensure that BCC is solvable.
We do not reproduce the lower bound proof here, but
in the rest of the paper, we assume that $n\geq (d+2)f+1$,
and also that $n\geq 2$ (because consensus is trivial when $n=1$).

\section{Preliminaries}
\label{s_ops}

Some notations introduced throughout the paper are summarized in Appendix \ref{app_s_notations}.
In this section,
we introduce operations $\HH$, $H_l$, $H$,
and two communication primitives, {\em reliable broadcast} and {\em stable vector}, used later in the paper.
\begin{definition}
\label{def:hh}
Given a set of points $X$, $\HH(X)$ is defined
as the convex hull of the points in $X$.
\end{definition}
\begin{definition}
\label{def:linear_hull}
Suppose that $\nu$ convex polytopes $h_1, h_2,\cdots, h_\nu$, and $\nu$ constants $c_1, c_2, \cdots, c_\nu$
are given such that
(i) $0 \leq c_i \leq 1$ and $\sum_{i = 1}^\nu c_i = 1$,
and (ii) for $1\leq i\leq\nu$, if $c_i\neq 0$, then $h_i\neq\emptyset$.
Linear combination of these convex
polytopes, $H_l(h_1, h_2,\cdots, h_\nu;~c_1, c_2, \cdots, c_\nu)$,
 is defined as follows:
\begin{itemize}
\item Let $Q := \{ i ~|~ c_i\neq 0, ~ 1\leq i\leq \nu \}$.
\item 
$p \in H_l(h_1, h_2,\cdots, h_\nu;~c_1, c_2,\cdots, c_\nu)$ if and only if
\begin{equation}
\label{eq:linear_hull}
\text{for each~} i\in Q,
\text{~there exists~} p_i \in h_i, ~~\text{such that}~~p = \sum_{i \in Q} c_i p_i
\end{equation}
\end{itemize}
\end{definition}
Note that a convex polytope may possibly consist of a single point.
Because $h_i$'s above are all convex, $H_l(h_1, h_2,\cdots, h_\nu;~c_1, c_2,\cdots, c_\nu)$ is also a
convex polytope (proof included in Appendix \ref{a:linear_convex} for completeness).
The parameters for $H_l$ consist of two lists, a list of polytopes
$h_1,\cdots,h_\nu$, and a list of weights $c_1,\dots,c_\nu$.  With a slight abuse of notation, we 
will specify one or both of these lists as either a {\em row vector}
or a {\em multiset}, with the understanding that the {\em row vector} or
{\em multiset} here represent an ordered list of its elements.

~

Function $H$ below is called in our algorithm with parameters $(\VV,t)$ wherein
$t$ is a round index ($t\geq 0$) and $\VV$ is a set of tuples of the
form $(h,j,t-1)$, where $j$ is a node identifier; when $t=0$, $h$ is a set of received messages in the previous round, and when $t>0$, $h$
is a convex polytope.

\vspace*{2pt}


\noindent
\hrule
{\bf Function $H(\VV,t)$}, $t\geq 0$:
\begin{list}{}{}
\item {\bf If $t=0$:}
\begin{itemize}
\item For each tuple $(x,k)$,
where $x$ is a point and $k$ is a node identifier,\\
 define $N(x,k) := |\{\,l\, |\, (\II,l,-1) \in \VV~~\text{and}~~(x,k,-1) \in \II\}|$.
\item Define set $Y:= \{\,(x,k) \,|\,N(x,k) \geq f+1\}$.
\item Define multiset $X := \{\,x \,|\,(x,k)\in Y\}$.
Size of multiset $X$ is identical to the size of set $Y$.
In a multiset, same element may appear multiple times.
\item
${\tt temp}~:= ~ \cap_{\,C \subseteq X, |C| = |X| - f}~~\HH(C)$.\\
The intersection above is over the convex hulls of the subsets of $X$ of size $|X| - f$.
\item Return {\tt temp}.
\end{itemize}

\item {\bf If $t>0$:}
\begin{itemize}
\item Define multiset $X:=\{h~|~ (h,j,t-1)\in \VV\}$.
In our use of function $H$, each $h\in X$ is always non-empty.

\item ${\tt temp}~:=~ H_l(X; \frac{1}{|X|}, \cdots,\frac{1}{|X|})$.
~~Note that all the weights here are equal to $\frac{1}{|X|}$.
\item Return {\tt temp}. 
\end{itemize}
\end{list}
\hrule

~

\comment{++++++++++++++++++++++++++++++++++++++++++++++++++++++++++++++
Function $H^h$ below is called in our algorithm with parameter $\VV$, which is a set of tuples of the
form $(\II,j,0)$, where $\II$ itself is a set of tuples of the form $(x_k,k,PR)$,\footnote{$x_k$ is the input at node $k$ if $k$ is fault-free.}, $j$ and $k$ are node identifiers, and $PR$ is a round index.

\vspace*{2pt}

\noindent
\hrule
{\bf Function $H^h(\VV)$}
\begin{enumerate}
\item For each point $x$, define $N(x) := |\{\,l\, |\, (\II,l,0) \in \VV~~\text{and}~~(x,k,PR) \in \II_l\}|$

+++++++++ removed subscript $k$, $l$ +++++++++++++

\item $Y:= \{\,(x, k, -1) \,|\,N(x) \geq f+1\}$

+++++++++++ Nitin changed 0 to $-1$ above ++++++++ and changed X to Y +++++

\item Return $H(Y, 0)$

\end{enumerate}
\hrule
~ \\
+++++++++++++++++++++++++++++++++++++++++++++++++++++++++++++++}

\paragraph{Communication Primitives:}

As seen later, our algorithm proceeds in asynchronous rounds. We label
the {\em preliminary round} as round $-1$,
and the remaining rounds as rounds 0, 1, 2, etc.
For communication between the nodes,
we use the {\em reliable broadcast} primitive
\cite{abraham_04_3t+1_async} and {\em stable vector} primitive \cite{renaming_JACM, herlihy_colorless_async}, which are also used in other related work \cite{herlihy_multi-dimension_AA, Vaidya_BVC, herlihy_colorless_async}. Note that we adopt the version of stable vector presented in \cite{herlihy_colorless_async}.
In particular, in round $-1$ (preliminary round), we use {\em stable vector}
and {\em reliable broadcast} both, as explained below. In rounds 0 and larger, we only use {\em reliable broadcast}.

\paragraph{Round $t$, $t\geq 0$:}
In round $t$, $t\geq 0$,
each node performs reliable broadcast of one message using $\RBSend$. 
Each message sent using $\RBSend$ consists of a 3-tuple of the form $(v,i,t)$: here, $i$ denotes the sender node's identifier, $t$ is round index, and $v$ is message value (the value $v$ itself is often a tuple). The operation $\RBSend(v,i,t)$ is used by node $i$ to perform {\em reliable broadcast} of $(v,i,t)$  in round $t$.
Each such message may be eventually reliably received by a fault-free node. 
When message $(v,j,t)$ is {\em reliably received} by some node $i$,
the event $\RBRecv(v,j,t)$ is said to have occurred at node $i$
 (note that $i$ may possibly be equal to $j$).
The second element in a reliably received 3-tuple message,
namely $j$ above, is always identical to the identifier
of the node that performed the corresponding reliable broadcast.
An appropriate handler is executed on each such $\RBRecv$ event,
as described in the algorithm.

\paragraph{Round $-1$:}
In round $-1$,
each node $i$ performs reliable broadcast of message
$(x_i,i, -1)$  using $\RBSend(x_i,i,-1)$, where $x_i$ is the input vector at node $i$.
The {\em stable vector} primitive $\SVRecv(-1)$ is then invoked via a blocking
call.
$\SVRecv(-1)$ eventually returns at each fault-free node with a set containing
at least $(n-f)$ messages of the form $(\cdot,\,\cdot,-1)$.
These sets have the desirable
property that the sets returned to {\em all} the fault-free nodes
contain at least $(n-f)$ messages in common.
Messages sent by some of the nodes using $\RBSend$ in round $-1$ may not be included
in the set returned to a fault-free node by $\SVRecv(-1)$. Each such
message may be later delivered to the fault-free node via a $\RBRecv$
event. Thus, for round $-1$, at fault-free node $i$, the $\RBRecv$
events may occur only for messages that are not returned by $\SVRecv$. 
An appropriate handler is executed on each such $\RBRecv$ event,
as described in the algorithm below.

With a slight abuse of terminology, when we say that node $j$ reliably receives $(v,i,-1)$, we mean that either
 (i) $(v,i,-1)$ is included the set returned by $\SVRecv(-1)$ to node $j$, or
 (ii) event $\RBRecv(v,i,-1)$ occurs at node $j$ after $\SVRecv(-1)$
	had already returned. 

\paragraph{Properties of Communication Primitives:}
Each fault-free node performs one reliable broadcast ($\RBSend$) in each round of our algorithm.
{\em Reliable broadcast} and {\em stable vector} achieve the properties listed below,
as proved previously \cite{abraham_04_3t+1_async,herlihy_colorless_async}.
In the properties below, round index $r\geq -1$.
\begin{itemize}
\item {\bf Fault-Free Integrity:} If a fault-free node $i$ {\em never} reliably broadcasts $(v, i, r)$,
then no fault-free node ever reliably receives $(v,i,r)$.

\item {\bf Fault-Free Liveness:} If a fault-free node $i$ performs reliable broadcast of $(v, i, r)$, then each fault-free node eventually reliably receives $(v, i, r)$.

\item {\bf Global Uniqueness:} If two fault-free nodes $i, j$ reliably receive $(v, k, r)$ and $(w, k, r)$, respectively, then $v = w$, even if node $k$ is faulty.

\item {\bf Global Liveness:} For any two fault-free nodes $i, j$, if $i$ reliably receives $(v, k, r)$, then $j$ will eventually reliably receive $(v, k, r)$, even if node $k$ is faulty.

\item {\bf Fault-free Containment}: For fault-free nodes $i,j$, let $R_i, R_j$ be the set of messages returned to nodes $i, j$ by stable vector
primitive $\SVRecv(-1)$
in round $-1$, respectively. Then,
$|R_i|\geq n-f$, $|R_j|\geq n-f$, and either $R_i \subseteq R_j$ or $R_j \subseteq R_i$.
\end{itemize}
The last above property ensures that, in round $-1$, all the fault-free nodes
receive at least $(n-f)$ identical messages.
In addition to the above property, the following property is also ensured:
\begin{itemize}
\item Any fault-free node $i$, for any $t$ and $j$,
reliably receives (either via $\SVRecv$ or $\RBRecv$) at most one message of
the form $(*,j,t)$.
\end{itemize}
This property is implemented easily by requiring each node to,
after receiving the first message of the form $(*,j,t)$, to simply ignore
any further messages of that form.
In our algorithm, each fault-free node $j$ reliably broadcasts exactly
one message of the form $(*,j,t)$ in any round $t$. Thus, the above property
is useful to avoid responding to multiple messages with the same round
index from a
faulty node.

\section{Proposed Algorithm: Optimal Verified Averaging}
\label{s_alg}

The proposed algorithm (named {\em Optimal Verified Averaging}) proceeds in asynchronous rounds.
The input at each node $i$ is a $d$-dimensional vector of reals, denoted as $x_i$. The
initial round is called a {\em preliminary round}, and also referred to as round $-1$.
Subsequent rounds are named round 0, 1, 2, etc.
In each round $t\geq 0$, each node $i$ computes a state variable $h_i$, which represents a convex polytope in the $d$-dimensional Euclidean space. We will refer to the value of $h_i$ at the {\em end} of the $t$-th round performed by node $i$ as $h_i[t]$, $t\geq 0$. Thus, for $t\geq 1$, $h_i[t-1]$ is the value of $h_i$ at the {\em start} of the $t$-th round at node $i$.


Similar to the algorithm in our prior work \cite{Tseng_BCC}, we use a technique named {\em verification}
to ensures that 
if a faulty node deviates from the algorithm specification
(except possibly choosing an invalid input vector),
then its incorrect messages will be ignored by
the fault-free nodes.
The {\em verification} mechanism is motivated by prior work by other researchers
\cite{Welch_textbook}.
With {\em verification},
aside from choosing a bad input, a faulty node cannot cause any other damage
to the execution.
 
Before we present the proposed algorithm, we introduce a convention for the brevity of presentation:
\begin{itemize}
\item When we say that $(*,i,t) \in \VV$, we mean that {\em there exists $z$ such that $(z,i,t)\in\VV$}.

\item When we say that $(*,i,t) \not\in \VV$, we mean that {\em $\forall z$, $(z,i,t)\notin\VV$}.
\end{itemize}
The proposed {\em Optimal Verified Averaging} algorithm for node $i\in V$ is presented below.
All references to line numbers in our discussion refer to numbers listed
on the right side of the algorithm pseudo-code.
Recall that in round $t\geq 0$,
whenever a message is reliably received by any node, a handler is called
to process that message. In round $-1$,
messages that are not delivered by $\SVRecv(-1)$ may be later 
reliably received via a $\RBRecv$ event, invoking the corresponding
handler. Multiple such handlers 
may execute {\em concurrently} at a given node. For correct behavior, line 7,
and lines 11-16 in the algorithm 
are atomically executed in a {\em critical section}. Thus,
even though multiple event handlers may execute simultaneously,
execution of line 7 in one instance of the handler is not interleaved with execution of any other handler instance;
similarly, 
execution of lines 11-16 in one instance of the handler is not interleaved with execution of any other handler instance.

\begin{itemize}
\item {\bf Round $-1$}:
In round $-1$, each node $i$ uses $\RBSend$ to reliably broadcast $(x_i,i,-1)$ where $x_i$ is its input (line 1). Each node then calls the primitive $\SVRecv(-1)$, which eventually returns with a set of messages tagged with index $-1$. These messages are stored in $\R_i[-1]$ and $\R^c_i[-1]$ both (lines 2 and 3). At this point, $\R^c_i[-1] = \R_i[-1]$.
At line 4, node $i$ also sets $h_i[-1]$ to be equal to a default value $\emptyset$ (because
$h_i[-1]$ does not affect future computations). Afterwards, each node can proceed to round 0 (line 5).

Note that reliable broadcast of a message by some node $j$ may not be received by node $i$ using $\SVRecv$ at line 2; however, the
message may be later reliably received by node $i$ via a $\RBRecv$ event (line 6).
Line 7 specifies the behavior of the event handler for event $\RBRecv(x,j,-1)$ at node $i$. Whenever a message of the form $(x, j, -1)$ is reliably received
via event $\RBRecv(x,j,-1)$ (line 6), the set $\R_i[-1]$ is updated (line 7). Since line 7 is performed atomically, $\R_i[-1]$ may continue to grow even after node $i$ has proceeded
to round 0; however, $\R^c_i[-1]$ is not modified again. Note that a message received by node $i$ via $\SVRecv$ (at line 2) or $\RBRecv$ (at line 6) may possibly have been reliably
broadcast by node $i$ itself.

\item {\bf Round $t\geq 0$:}
In round $t\geq 0$, {\em Optimal Verified Averaging} adopts a similar structure to
round $-1$, with one key difference: stable vector ($\SVRecv$) is not used in these rounds, and all messages are received via $\RBRecv$ events.
In round $t\geq 0$, node $i$ first reliably broadcasts
message $((h_i[t-1],\R^c_i[t-1]),i,t)$ (line 8). 
Lines 9-16 specify the event handler for event $\RBRecv((h,\VV),j,t)$ at node $i$. 
Whenever a message of the form $((h,\VV),j,t)$ is reliably received from node $j$ (line 9),
node $i$ first waits until its own set $\R_i[t-1]$ becomes
large enough to contain $\VV$. Note that $\R_i[t-1]$ is initially computed
in the round $t-1$, but it may continue to grow even after node
$i$ proceeds to round $t$.
If the condition $\VV\subseteq\R_i[t-1]$ never becomes true, then this
message is not processed further.

Recall that lines 11-16 are performed atomically.
The message $((h,\VV),j,t)$ is considered to be
{\em verified} if Procedure $\Verify(h,\VV,j,t)$ returns \TRUE (line 11).
As shown in the pseudo-code for $\Verify$, the verification checks performed 
are different for $t=0$, $t=1$ and $t\geq 2$. 
If a message is thus verified, then some elements in the message are added to $\R_i[t]$ via Procedure $\Add(\cdot)$ (line 12). As shown in the pseudo-code for $\Add$, the elements added are different for $t=0$, and $t\geq 1$.

Procedure $\Proceed(t)$ at line 13 determines whether set $\R_i[t]$ has grown to
a point where it is appropriate to compute the new state $h_i[t]$ (line 15) and
proceed to round $t+1$ (line 16).
The checks performed in $\Proceed(t)$ are different for $t=0$ and $t\geq 1$,
as shown in the pseudo-code for $\Proceed$.
The value of $\R_i[t]$ used to compute $h_i[t]$ is stored in
$\R_i^c[t]$ (line 14).

New messages may still be added to $\R_i[t]$ if
events of the form $\RBRecv((h,\VV),j,t)$ occur
after node $i$ has proceeded to round $t+1$.
Thus, $\R_i[t]$ may continue to grow even after node $i$ has proceeded
to round 1; however, $\R^c_i[t]$ is not modified again, and remains unchanged
after it is set at line 14.

\end{itemize}

\vspace*{8pt}
\hrule

\vspace*{2pt}

\noindent {\bf Optimal Verified Averaging Algorithm:} Steps performed at node $i$ shown below.\\~\\
The algorithm terminates after $t_{end}$ rounds, where $t_{end}$ is a constant, defined in (\ref{e_end}).

\vspace*{4pt}

\hrule

\vspace*{8pt}

\noindent {\bf Initialization:} All sets used below are initialized to $\emptyset$.\\


\noindent {\bf Preliminary Round (Round $-1$) at node $i$:} 	

\begin{itemize}
\item $\RBSend(x_i,i,-1)$ 		\hfill 1

\item $\R_i[-1] := \SVRecv(-1)$ \hfill 2
\item $\R^c_i[-1] := \R_i[-1]$ \hfill 3

\item $h_i[-1] := \emptyset$  \hfill 4
\item Proceed to Round 0 \hfill 5

~

\indent {\em Comment}: Message sent by $j$ using $\RBSend$ may not be received by $i$ using $\SVRecv$ at line 2.\\
 \indent Due to Fault-free Liveness property of the primitive, this message will later be received \\
\indent by $i$ using $\RBRecv$ at line 6 below.\\

\item
{\bf Event handler for event $\RBRecv(x, j, -1)$ at node $i$} : \hfill 6\\

\centerline{\underline{\em Line 7 is performed atomically.}}

\begin{list}{}{}
\item[\hspace*{0.4in}$-$] $\R_i[-1] := \R_i[-1] \cup \{(x,j,-1)\}$ \hfill 7 \\

\end{list}

\end{itemize}

\noindent {\bf Round $t\geq 0$ at node $i$:} 

\begin{itemize}
\item $\RBSend( (h_i[t-1],\R^c_i[t-1]), i, 0)$ \hfill 8 \\

\item
{\bf Event handler for event $\RBRecv((h,\VV), j, t)$ at node $i$} : \hfill 9

\begin{itemize}

\item Wait until $\VV \subseteq \R_i[t-1]$ \hfill 10 \\

\centerline{\underline{\em Lines 11-16 are performed atomically.}}

\item If~~ $\Verify(h,\VV,j,t)$ returns \TRUE then \hfill 11

	\hspace*{1.6in} $\R_i[t] := \Add(\R_i[t], h, \VV, j, t)$ \hfill 12 \\

\item When $\Proceed(t)$ returns \TRUE \underline{for the first time} \hfill 13 \\
  \hspace*{1.6in} $\R^c_i[t] := \R_i[t]$ \hfill 14

	\hspace*{1.6in} $h_i[t] := H(\R^c_i[t],t)$ \hfill 15

	\hspace*{1.6in} Proceed to Round $t+1$  \hfill 16

\end{itemize}
\end{itemize}
\hrule

~

Procedure $\Verify(h,\VV,j,t)$ at node $i$:
\begin{itemize}
\item Case $t=0$: If $|\VV|\geq n-f$, then return \TRUE, else return \FALSE.
\item Case $t=1$: If $|\VV|\geq n-f$ and $h=H(\VV,0)$, then return \TRUE, else return \FALSE.
\item Case $t\geq 2$: If $|\VV|\geq n-f$ and $h=H(\VV,t-1)$ and $(*,j,t-2)\in\VV$,\\
 \hspace*{1in}then return \TRUE, else return \FALSE.
\end{itemize}

\hrule

~

Procedure $\Add(\R_i[t], h,\VV,j,t)$ at node $i$:
\begin{itemize}
\item Case $t=0$: return $\R_i[t] \cup \{(\VV,j,-1)\}$.
\item Case $t\geq 1$: return $\R_i[t] \cup \{(h,j,t-1)\}$.

\end{itemize}

\hrule

~

Procedure $\Proceed(t)$ at node $i$:
\begin{itemize}
\item Case $t=0$:
if $|\R_i[0]|\geq n-f$,\\
 \hspace*{1in}then return \TRUE,\\ \hspace*{1in}else return \FALSE. 
\item Case $t\geq 1$:
if $|\R_i[t]|\geq n-f$ and $(h_i[t-1],i,t-1)\in\R_i[t]$,\\
 \hspace*{1in}then return \TRUE,\\ \hspace*{1in}else return \FALSE. 
\end{itemize}
\hrule

\vspace*{8pt}

~

The algorithm terminates after $t_{end}$ rounds, where $t_{end}$ is a constant, defined in (\ref{e_end}). The state $\bfv_i[t_{end}]$ of each node $i$ is its output when the algorithm terminates after $t_{end}$ iterations.

\begin{definition}
\label{d_verified}
A node $k$'s execution of round $r$, $r\geq 0$, is said to be
\underline{\bf verified by a fault-free node $i$} if, eventually node $i$ reliably receives message
of the form $((h,\VV),k,r+1)$ from node $k$, and subsequently adds $(h,k,r)$ to $\R_i[r+1]$
(at line 12). 
Note that node $k$ may possibly be faulty.
Node $k$'s execution of round $r$ is said to be \underline{\bf verified} if it
is verified by at least one fault-free node.
\end{definition}
We now introduce some more notations (which are also summarized in Appendix \ref{app_s_notations}):
\begin{itemize}
\item For a {\em given} execution of the proposed algorithm, let $F$ denote the {\em actual} set of faulty nodes in the execution.
Let $|F|=\phi$. Thus, $0\leq \phi\leq f$.
\item For $r\geq 0$, let $F_v[r]$ denote the set of faulty nodes whose round $r$ execution is verified by at least one fault-free node, as per Definition \ref{d_verified}.
Note that $F_v[r]\subseteq F$.
\item Define $\overline{F_v}[r]=F-F_v[r]$, for $r\geq 0$.
\end{itemize}
For each faulty node $k\in F_v[r]$, by
	Definition \ref{d_verified}, there must exist
	a fault-free node $i$ that eventually reliably receives
	a message of the form $((h,\VV),k,r+1)$ from
	node $k$, and adds $(h,k,r)$ to $\R_i[r+1]$.
	Given these $h$ and $\VV$, for future reference, let us define 
	\begin{eqnarray}
	 h_k[r]&=&h \label{e_faulty_h} \\
	 \R^c_k[r]&=&\VV \label{e_faulty_V}
	\end{eqnarray}
	Node $i$ verifies node $k$'s round $r$ execution after node $i$
	has entered its round $r+1$.
	Since round $r$ execution of faulty node $k$ above is verified by fault-free
	node $i$, due to the checks performed in procedure $\Verify$, the equality below holds
	for $h_k[r]$ and $\R^c_k[r]$ defined in (\ref{e_faulty_h}) and (\ref{e_faulty_V}).

	\begin{eqnarray}
   h_k[r]& = & H(\R^c_k[r],r)~~~\text{for}~~r \geq 0 \label{e_faulty_H}
	\end{eqnarray}

(The proof of Claim \ref{c_ver} in Appendix \ref{a_claims} elaborates on the
above equality.)
	While the algorithm requires each node $k$ to maintain
	variables $h_k[r]$ and $\R^c_k[r]$, we cannot assume correct behavior on
	the part of the faulty nodes. However, from the perspective of
	each fault-free node that verifies the round $r$ execution of faulty node $k\in F_v[r]$,
	node $k$ behaves ``as if'' these local variable take the values specified
	in (\ref{e_faulty_h}) and (\ref{e_faulty_V}) that satisfy (\ref{e_faulty_H}). 
	Note that if the round $r$ execution (where $r\geq 0$) of a faulty node $k$ is verified by more than
	one fault-free node, due to the {\em Global Uniqueness} of reliable broadcast,
	all these fault-free nodes must have reliably received identical round $r+1$ messages
	from node $k$.  

Proofs of Lemmas \ref{l_progress}, \ref{lemma:J_in_H0} and \ref{lemma:always_notFv_not_verified}
below
are presented in Appendices \ref{a_lemma_progress}, \ref{a_lemma_J}, and
\ref{a_always}, respectively.
These lemmas are used to prove the correctness of the {\em Optimal Verified Averaging} algorithm.

\begin{lemma}
\label{l_progress}
Optimal Verified Averaging ensures {\em progress}: (i) all the fault-free nodes will eventually progress to round 0; and, (ii)
if all the fault-free nodes progress to the start of round $t$, $t \geq 0$, then all the
fault-free nodes will eventually progress to the start of round $t+1$.
\end{lemma}

\begin{lemma}
\label{lemma:J_in_H0}
For each node $i\in V-\overline{F_v}[0]$, the polytope $h_i[0]$ is non-empty.
\end{lemma}

~

\begin{lemma}
\label{lemma:always_notFv_not_verified}
For $r\geq 0$,
if $b \in \overline{F_v}[r]$, then for all $\tau\geq r$,
\begin{itemize}
\item $b \in \overline{F_v}[\tau]$, and 
\item for all $i\in V-\overline{F_v}[\tau+1]$, $(*,b,\tau)\not\in\R^c_i[\tau+1]$.
\end{itemize}
\end{lemma}

\newcommand{\printcomment}[1]{{\bf #1}}

\section{Correctness}
\label{s_correctness}

We first
introduce some terminology and definitions related to matrices. Then, we develop a
{\em transition matrix} representation of the proposed algorithm, and use
that to prove its correctness. Note that the technique is identical to the one present in our prior work \cite{Tseng_BCC}. We include the proof here for completeness.

\subsection{Matrix Preliminaries}

We use boldface upper case letters to denote matrices, rows of matrices, and their elements. For instance, $\bfA$ denotes a matrix, $\bfA_i$ denotes the $i$-th row of matrix $\bfA$, and $\bfA_{ij}$ denotes the element at the intersection of the $i$-th row and the $j$-th column of matrix $\bfA$.

\begin{definition}
\label{d_stochastic}
A vector is said to be stochastic if all its elements
are non-negative, and the elements add up to 1.
A matrix is said to be row stochastic if each row of the matrix is a
stochastic vector. 
\end{definition}
For matrix products, we adopt the ``backward'' product convention below, where $a \leq b$,
\begin{equation}
\label{backward}
\Pi_{\tau=a}^b \bfA[\tau] = \bfA[b]\bfA[b-1]\cdots\bfA[a]
\end{equation}
For a row stochastic matrix $\bfA$,
 coefficients of ergodicity $\delta(\bfA)$ and $\lambda(\bfA)$ are defined as
follows \cite{Wolfowitz}:
\begin{eqnarray*}
\delta(\bfA) & = &   \max_j ~ \max_{i_1,i_2}~ \| \bfA_{i_1\,j}-\bfA_{i_2\,j} \| \label{e_zelta} \\
\lambda(\bfA) & = & 1 - \min_{i_1,i_2} \sum_j \min(\bfA_{i_1\,j} ~, \bfA_{i_2\,j}) \label{e_lambda}
\end{eqnarray*}
\begin{claim}
\label{claim_zelta}
For any $p$ square row stochastic matrices $\bfA(1),\bfA(2),\dots, \bfA(p)$, 
\begin{eqnarray*}
\delta(\Pi_{\tau=1}^p \bfA(\tau)) ~\leq ~
 \Pi_{\tau=1}^p ~ \lambda(\bfA(\tau)).
\end{eqnarray*}
\end{claim}
Claim \ref{claim_zelta} is proved in \cite{Hajnal58}.
%
%
%
%
Claim \ref{c_lambda_bound} below follows directly from the definition of $\lambda(\cdotp)$. 
\begin{claim}
\label{c_lambda_bound}
If there exists a constant $\gamma$, where $0<\gamma\leq 1$, such
that, for any pair of rows $i,j$ of matrix $\bfA$, there exists a column
$g$ (that may depend on $i,j$) such that,
$\min (\bfA_{ig},\bfA_{jg}) \geq \gamma$,
then
 $\lambda(\bfA)\leq 1-\gamma<1$.
\end{claim}

Let $\vectorv$ be a column vector with $n$ elements, such that the $i$-th element of vector $\vectorv$, namely $\vectorv_i$, is a convex polytope in the $d$-dimensional Euclidean space. Let $\bfA$ be a $n\times n$ row stochastic square matrix. Then multiplication of
matrix $\bfA$ and vector $\vectorv$ is performed by multiplying each row of $\bfA$
with column vector $\vectorv$ of polytopes. Formally, 
\begin{equation}
\label{eq:multiplication}
\bfA \vectorv = [H_l(\vectorv^T; \bfA_1)
~~~~~H_l(\vectorv^T; \bfA_2)
~~~~~...~~~~~
H_l(\vectorv^T; \bfA_n)]^T
\end{equation}
where $^T$ denotes the transpose operation (thus, $\vectorv^T$ is the transpose of $\vectorv$).
$H_l$ is defined in Definition \ref{def:linear_hull}. Thus, the result of the multiplication
$\bfA\vectorv$ is a column vector consisting of $n$ convex polytopes.
Similarly, product of row vector $\bfA_i$ and above vector $\vectorv$ is
obtained as follows, and it is a polytope.
\begin{eqnarray}
\label{e_r_c}
\bfA_i\vectorv &=& H_l(\vectorv^T\,;~\bfA_i)
\end{eqnarray}

\subsection{Transition Matrix Representation of {\em Optimal Verified Averaging}}

\noindent
Let $\vectorv[t]$, $t\geq 0$, denote a column vector of length $|V|=n$.
In the remaining discussion,
we will refer to $\vectorv[t]$ as the state of the system at the end of round $t$.
In particular,
$\vectorv_i[t]$ for $i\in V$ is viewed as
the state of node $i$ at the end of round $t$.
We define $\vectorv[0]$ as follows:
\begin{itemize}
\item[(I1)] For each fault-free node $i\in V-F$, $\vectorv_i[0]:=h_i[0]$.
\item[(I2)] For each faulty node $k\in F_v[0]$, 
	$\vectorv_k[0]:=h_k[0]$, where
	$h_k[0]$ is defined in (\ref{e_faulty_h}).
\item[(I3)] For each faulty node $k\in \overline{F_v}[0]$,
$\vectorv_k[0]$ is {\em arbitrarily} defined as the origin in the $d$-dimensional Euclidean space.
We will justify this arbitrary choice later.
\end{itemize}

We will show that the state evolution 
can be represented in a matrix form as in (\ref{matrix:alg1}) below, for a suitably
chosen $n\times n$ matrix  $\matrixm[t]$.
$\matrixm[t]$ is said to be the {\em transition matrix} for round $t$.
\begin{equation}
\label{matrix:alg1}
\vectorv[t] = \matrixm[t]~\vectorv[t-1], ~~~~~ t\geq 1
\end{equation}
For all $t\geq 0$,
Theorem \ref{t_M} below proves that,
for each $i\in V-\overline{F_v}[t]$, $h_i[t]=\vectorv_i[t]$.

Given a particular execution of the algorithm, we construct the transition matrix
$\bfM[t]$
for round $t\geq 1$ using
the following procedure.

\vspace*{10pt}
\hrule

\vspace*{2pt}

\noindent {\bf Construction of the Transition Matrix for Round $t~(t \geq 1)$}

\vspace*{4pt}

\hrule

\vspace*{8pt}

\begin{itemize}
\item For each node $i \in V-\overline{F_v}[t]$, and each $k\in V$:

\begin{list}{}{}
\item{} If $(*,k,t-1) \in \R^c_i[t]$, then 

\begin{equation}
\label{eq:matrix_i}
\matrixm_{ik}[t] := \frac{1}{|\R^c_i[t]|}
\end{equation}

\item{} Otherwise,
\begin{equation}
\label{eq:matrix_i-2}
\matrixm_{ik}[t] := 0
\end{equation}

\end{list}

{\em Comment:}
For a faulty node $i\in F_v[t]$, $h_i[t]$ and $\R^c_i[t]$
are defined in (\ref{e_faulty_h}) and (\ref{e_faulty_V}).\\

\item For each node $j \in \overline{F_v}[t]$, 
and each $k \in V$,
\begin{eqnarray}
\matrixm_{jk}[t] &:=& \frac{1}{n}
\label{e_fv}
\end{eqnarray}

\end{itemize}

\hrule

\vspace*{8pt}

~

\begin{theorem}
\label{t_M}
For $r\geq 0$,
with state evolution
specified as $\vectorv[r+1]=\bfM[r+1]\vectorv[r]$ using $\bfM[r+1]$ constructed above,
for all $i\in V-\overline{F_v}[r]$,
(i) $h_i[r]$ is non-empty, and
 (ii) $h_i[r]=\vectorv_i[r]$.
\end{theorem}
\begin{proof}

The proof of the theorem is by induction. The theorem holds for $r=0$ due to
Lemma \ref{lemma:J_in_H0}, and the choice
of the elements of $\vectorv[0]$, as specified in (I1), (I2) and (I3) above.

Now, suppose that the theorem holds for $r=t-1$ where $t-1\geq 0$,
and prove it for $r=t$. 
Thus, by induction hypothesis, for
all $i\in V-\overline{F_v}[t-1]$, $h_i[t-1]=\vectorv_i[t-1]\neq \emptyset$.
Now, $\vectorv[t] = \bfM[t]\vectorv[t-1]$.
\begin{itemize}
\item In round $t\geq 1$, each fault-free node $i\in V-F$ computes its new state 
$h_i[t]$
at line 15 using function $H(\R^c_i[t],t)$. The function $H(\R^c_i[t],t)$ for $t\geq 1$ then computes a
linear combination of $|\R^c_i[t]|$ convex hulls, with all the weights being equal to
$\frac{1}{|\R^c_i[t]|}$.
Also,
by Definition \ref{d_verified} and the definition of $\overline{F_v}[t-1]$, if $(h,j,t-1)\in \R^c_i[t]$, then $j\not\in \overline{F_v}[t-1]$ (i.e., $j\in V-\overline{F_v}[t-1]$).
Therefore, if $(h,j,t-1)\in \R^c_i[t]$, then either $j$ is fault-free, or it is faulty and its
round $t-1$ execution is verified: thus, $h=h_j[t-1]$.
Also, by induction hypothesis, $h=h_j[t-1]\neq\emptyset$.
This implies that $h_i[t]=H(\R^c_i[t],t)$ is non-empty.

Then observe that, by defining $\bfM_{ik}[t]$ elements as in (\ref{eq:matrix_i})
and (\ref{eq:matrix_i-2}), we ensure that
$\bfM_i[t]\vectorv[t-1]$ equals $H(\R^c_i[t],t)$, and hence equals $h_i[t]$.

\item
For $i\in F_v[t]$ as well, as shown in (\ref{e_faulty_H}), $h_i[t]=H(\R^c_i[t],t)$, where
$h_i[t]$ and $\R^c_i[t]$ are as defined in (\ref{e_faulty_h})
and (\ref{e_faulty_V}).
The function $H(\R^c_i[t],t)$ for $t\geq 1$ then computes a
linear combination of $|\R^c_i[t]|$ convex hulls, with all the weights being equal to
$\frac{1}{|\R^c_i[t]|}$.
Consider an element $(h,j,t-1)$ in $\R^c_i[t]$. We argue that $j \in V - \overline{F_v}[t-1]$. Suppose this is not true, i.e., $j \in \overline{F_v}[t-1]$. 
By Definition \ref{d_verified}, node $i$'s round $t$ execution is verified by some fault-free node $k$, which implies that eventually, $\R^c_i[t] \subseteq \R_k[t]$. However, since $k$ is fault-free, and $(h,j,t-1) \not\in \R_k[t]$, a contradiction.
Hence, if $(h,j,t-1)\in \R^c_i[t]$, then $j\in V-\overline{F_v}[t-1]$.
That is, if $(h,j,t-1)\in \R^c_i[t]$, then either $j$ is fault-free, or it is faulty and its
round $t-1$ execution is verified: thus, $h=h_j[t-1]$.

Also, by induction hypothesis, $h=h_j[t-1]\neq\emptyset$.
This implies that $h_i[t]=H(\R^c_i[t],t)$ is non-empty.

Then observe that, by defining $\bfM_{ik}[t]$ elements as in (\ref{eq:matrix_i})
and (\ref{eq:matrix_i-2}), we ensure that
$\bfM_i[t]\vectorv[t-1]$ equals $H(\R^c_i[t],t)$, and hence equals $h_i[t]$.

\end{itemize}
\end{proof}


Now, we argue that for $t \geq 0$, the state $\vectorv_j[t]$ for each node $j\in\overline{F_v}[t]$ does not affect the state
of the nodes $V-\overline{F_v}[\tau]$, for $\tau\geq t+1$.
From the discussion in the above proof, we see that for $j\in\overline{F_v}[t]$, $(*,j,t)\not\in \R^c_i[t+1]$ for $i \in V-\overline{F_v}[t+1]$. Thus, the sate $\vectorv_j[t]$ does not affect the state $h_i[t+1]$. Then, by Lemma \ref{lemma:always_notFv_not_verified}, if $j\in\overline{F_v}[t]$, then $j\in\overline{F_v}[\tau]$, for $\tau\geq t+1$. Thus, by the same argument, the sate $\vectorv_j[\tau]$ does not affect the state $h_i[t+1]$.
This justifies the somewhat arbitrary choice of
$\vectorv_j[0]$ for $j\in\overline{F_v}[0]$,
and $\matrixm_{jk}[t]$ in (\ref{e_fv}) for $j\in \overline{F_v}[t], ~t\geq 1$.
This choice does simplify the remaining proof somewhat.

The above discussion shows that, for $t\geq 1$, the evolution of $\vectorv[t]$ can be written as
in (\ref{matrix:alg1}), that is, $\vectorv[t]=\matrixm[t]\vectorv[t-1]$.
Given the matrix product definition in (\ref{eq:multiplication}),
it is easy to verify that
\[ \bfM[\tau+1] ~ \left(\bfM[\tau] \vectorv[\tau-1]\right) ~=~ \left(\bfM[\tau+1]\bfM[\tau]\right)~\vectorv[\tau-1]
\text{~ for~} \tau \geq 1.
\]
Therefore, by repeated application of (\ref{matrix:alg1}),
we obtain:
\begin{eqnarray}
\vectorv[t] & = & \left(\,\Pi_{\tau=1}^t \matrixm[\tau]\,\right)\, \vectorv[0],
 ~~~~ t\geq 1
\label{e_unroll}
\end{eqnarray}
Recall that we adopt the ``backward'' matrix product convention presented in
(\ref{backward}).

\begin{lemma}
\label{lemma:transition_matrix}
For $t\geq 1$,
transition matrix $\bfM[t]$ constructed using the above procedure satisfies the following conditions. 
\begin{itemize}
\item For $i,j \in V$, there exists a fault-free node $g(i,j)$ such that $\bfM_{ig(i,j)}[t] \geq \frac{1}{n}$. 

\item $\bfM[t]$ is a row stochastic matrix, and $\lambda(\bfM[t])\leq 1-\frac{1}{n}$.


\end{itemize}

\end{lemma}
The proof of Lemma \ref{lemma:transition_matrix} is presented in Appendix \ref{a_t_matrix}.

\subsection{Correctness of {\em Optimal Verified Averaging}}

\begin{definition}
\label{def:valid_hull}
A convex polytope $h$ is said to be {\em valid} if every point in $h$ is in the convex hull of the inputs at the fault-free nodes. 
\end{definition}

Lemmas \ref{lemma:valid_initial_hull}
and \ref{lemma:linear_valid} below are proved in  
Appendices \ref{app_s_lemma:valid_initial_hull} and \ref{app_s_lemma:linear_valid}, respectively.

\begin{lemma}
\label{lemma:valid_initial_hull}
$h_i[0]$ for each node $i \in \sv - \overline{F_v}[0]$ is valid.
\end{lemma}

\begin{lemma}
\label{lemma:linear_valid}
Suppose non-empty convex polytopes $h_1, h_2, \cdots, h_k$ are all valid. Consider $k$ constants $c_1, c_2, \cdots, c_k$ such that $0 \leq c_i \leq 1$ and $\sum_{i = 1}^k c_i = 1$.
Then the linear combination of these convex polytopes,
$H_l(h_1, h_2, \cdots, h_k; c_1, c_2, \cdots, c_k)$, is valid.
\end{lemma}

\begin{theorem}
\label{thm:correctness}

Optimal Verified Averaging satisfies the {\em validity},
{\em $\epsilon$-agreement} and {\em termination} properties
after a large enough number of asynchronous rounds.
\end{theorem}

\begin{proof}
Repeated applications of
Lemma \ref{l_progress} ensures that the fault-free nodes will progress
from the preliminary round through round $r$, for any $r\geq 0$, allowing us to use (\ref{e_unroll}).
Consider round $t \geq 1$. Let
\begin{eqnarray}\bfM^* &=& \Pi_{\tau=1}^t \bfM[\tau]
\label{e_Mstar}
\end{eqnarray}
(To simplify the presentation, we do not include the round index $[t]$
in the notation $\bfM^*$ above.)
 Then $\vectorv[t]=\bfM^* \vectorv[0]$.
By Lemma \ref{lemma:transition_matrix},
each $\bfM[t]$ is a {\em row stochastic} matrix,
therefore, $\bfM^*$ is also row stochastic. 
By Lemma \ref{lemma:valid_initial_hull},
$h_i[0]=\vectorv_i[0]$ for each $i\in \sv - \overline{F_v}[0]$ is valid. Therefore,
by Lemma \ref{lemma:linear_valid}, $\bfM_i^*\vectorv[0]$ for each $i\in V-F$
is valid. Also, by Theorem \ref{t_M} and (\ref{e_unroll}), $h_i[t]=\bfM_i^*\vectorv[0]$
for $i\in V-F$.
Thus, $h_i[t]$ is valid for $t \geq 1$. This observation together with Lemma \ref{lemma:valid_initial_hull} implies that {\em Optimal Verified Averaging} satisfies the validity condition for all round $r\geq 0$.

Let us define $\alpha = 1-\frac{1}{n}$.
By Lemma \ref{lemma:transition_matrix}, $\lambda(\bfM[t]) \leq 1-\frac{1}{n}=\alpha$. Then by Claim \ref{claim_zelta},
\begin{eqnarray}
\label{e_alpha}
\delta(\bfM^*)=
\delta(\Pi_{\tau=1}^t \bfM[\tau])
~ \leq ~ \lim_{t\rightarrow\infty} \Pi_{\tau=1}^{t} \lambda(\bfM[\tau]) 
 ~\leq~ {\left(1-\frac{1}{n}\right)}^t ~=~\alpha^t
\end{eqnarray}
%
%

Consider any two fault-free nodes $i,j\in V-F$.
By (\ref{e_alpha}), $\delta(\bfM^*)\leq \alpha^t$. 
Therefore, by the definition of $\delta(\cdot)$, for $1 \leq k \leq n$,
\begin{equation}
\label{eq:delta1}
\| \bfM^*_{ik} -  \bfM^*_{jk}\| \leq \alpha^t
\end{equation}
%
%
By Lemma \ref{lemma:always_notFv_not_verified}, and construction of the transition matrices,
it should be easy to see that $\bfM^*_{ib}=0$ for $b\in \overline{F_v}[0]$.
Then, for
any point $p_i^*$ in $h_i[t]=\bfM^*_i\vectorv[0]$,
there must exist, for all $k\in V-\overline{F_v}[0]$, $p_k \in h_k[0],$ such that
\begin{equation}
\label{pi}
p_i^* = \sum_{k\in V-\overline{F_v}[0]} \bfM^*_{ik} p_k
= \left(\sum_{k\in V-\overline{F_v}[0]} \bfM^*_{ik} p_k(1),~~\sum_{k\in V-\overline{F_v}[0]} \bfM^*_{ik} p_k(2), \cdots, \sum_{k\in V-\overline{F_v}[0]} \bfM^*_{ik} p_k(d)\right)
\end{equation}
where $p_k(l)$ denotes the value of $p_k$'s $l$-th coordinate.
Now choose point $p_j^*$ in $h_j[t]$ defined as follows.
\begin{equation}
\label{pj}
p_j^* = \sum_{k\in V-\overline{F_v}[0]} \bfM^*_{jk} p_k 
= \left(\sum_{k\in V-\overline{F_v}[0]} \bfM^*_{jk} p_k(1),~~ \sum_{k\in V-\overline{F_v}[0]} \bfM^*_{jk} p_k(2), \cdots,  \sum_{k\in V-\overline{F_v}[0]} \bfM^*_{jk} p_k(d)\right)
\end{equation}

Then the Euclidean distance between $p_i^*$ and $p_j^*$ is $d(p_i^*,p_j^*)$.
The following derivation is obtained by simple algebraic manipulation, using (\ref{eq:delta1}),
(\ref{pi}) and (\ref{pj}).
The omitted steps in the algebraic manipulation are shown in Appendix \ref{a_hausdorff}.
\begin{eqnarray}
d(p_i^*, p_j^*) &=  \sqrt{\sum_{l=1}^d (p_i^*(l) - p_j^*(l))^2} = \sqrt{\sum_{l=1}^d \left(\sum_{k\in V-\overline{F_v}[0]} \bfM^*_{ik} p_k(l) - \sum_{k\in V-\overline{F_v}[0]} \bfM^*_{jk} p_k(l)\right)^2}\nonumber\\
&\leq \alpha^t \sqrt{\sum_{l=1}^d \left(\sum_{k\in V-\overline{F_v}[0]} \|p_k(l)\|\right)^2}\label{eq:d_pipj1} \leq ~~ \alpha^t\Omega
\label{e_o}
\end{eqnarray} 
where
$\Omega = \max_{p_k \in h_k[0],k\in V-\overline{F_v}[0]} \sqrt{\sum_{l=1}^d (\sum_{k\in V-\overline{F_v}[0]} \|p_k(l)\|)^2}$.
Because the $h_k[0]$'s in the definition of $\Omega$ are all valid (by
Lemma \ref{lemma:valid_initial_hull}),
$\Omega$ can itself be upper bounded by a function of the input vectors at the fault-free
nodes.
In particular, under the assumption that each element of fault-free nodes' input vectors is
upper bounded by $U$ and lower bounded by $\mu$, $\Omega$ is upper
bounded by $\sqrt{dn^2\max(U^2,\mu^2)}$.
Observe that the upper bound on the right side of (\ref{e_o}) monotonically decreases with $t$, because $\alpha<1$.
Define $t_{end}$ as the smallest positive integer $t$
for which
\begin{eqnarray}\alpha^t\sqrt{dn^2\max(U^2,\mu^2)}<\epsilon \label{e_end} \end{eqnarray}
Recall that the algorithm terminates after $t_{end}$ rounds.
(\ref{e_o}) and (\ref{e_end}) together imply that, 
for fault-free $i,j$, for each point $p_i^*\in h_i[t_{end}]$
there exists a point $p_j^*[t]\in h_j[t_{end}]$ such that $d(p^*_i,p^*_j)<\epsilon$ (and, similarly, vice-versa).
Thus, by Definition \ref{def:dist}, 
Hausdorff distance $\D(h_i[t_{end}],h_j[t_{end}]) <\epsilon$.
Since this holds true for any pair of fault-free nodes
$i,j$, the $\epsilon$-agreement property is satisfied at termination.
\end{proof}

\section{Optimality of {\em Optimal Verified Averaging}}
\label{s_size}

Due to the {\em Fault-free Containment} property of {\em Stable Vector}, all fault-free nodes share at least $(n-f)$ messages in $\R^c_i[-1]$ (see lines 2-3). Let $Z$ denote the set of these shared messages, that is,
\begin{eqnarray}
Z &:=& \cap_{j\in V-F} \R^c_j[-1]
\label{e_Z}
\end{eqnarray}
Define $X_Z =: \{x~|~(x, k, -1) \in Z\}$. Then, define a convex polytope $I_Z$ as follows.
\begin{eqnarray}
I_Z& := &\cap_{D \subset X_Z, |D| = |X_Z|-f}\, \HH(D)
\label{e_I_Z}
\end{eqnarray}

The following lemma establishes a ``lower bound'' on the convex polytope that the fault-free nodes decide on. Recall that $\overline{F_v}[0]$ is defined as all the faulty nodes that are not verified by any fault-free nodes in Round 1. The proof is presented in Appendix \ref{a_l:svSize}.

\begin{lemma}
\label{lemma:svSize}
For all $i\in V-\overline{F_v}[t]$ and $t \geq 0$,
 $I_Z \subseteq h_i[t]$.
\end{lemma}

Then, the following key theorem shows that the presented algorithm is optimal. The proof is presented in Appendix \ref{a_t:optSize}. This theorem closes an open question raised in our prior work \cite{Tseng_BCC}.

\begin{theorem}
\label{thm:optSize}
The output convex polytope at fault-free node $i$ using {\em Optimal Verified Averaging} is optimal as per Definition \ref{def:optSize}.
\end{theorem}

\section{Summary}

This paper considers Byzantine Convex Consensus (BCC), wherein each node has a $d$-dimensional vector as its input, and each fault-free node should agree on an output polytope that is in the convex hull of
the input vectors at the fault-free nodes. We present an asynchronous approximate BCC algorithm
with optimal fault tolerance that reaches consensus
on an {\em optimal} output polytope.



\clearpage
\appendix

\section{Notations}
\label{app_s_notations}

This appendix summarizes some of the notations and terminology introduced throughout the paper.

\begin{itemize}
\item $n =$ number of nodes. We assume that $n\geq 2$.
\item $f =$ maximum number of Byzantine nodes.
\item $\sv = \{1, 2, \cdots, n\}$ is the set of all nodes.
\item $d = $ dimension of the input vector at each node.
\item $d(p, q) = $ the function returns the Euclidean distance between points $p$ and $q$.
\item $d_H(h_1, h_2) = $ the Hausdorff distance between convex polytopes $h_1, h_2$.
\item $\HH(C) = $ the convex hull of a multiset $C$.
\item $H_l(h_1, h_2, \cdots, h_k;~c_1, c_2, \cdots, c_k)$, defined in Section \ref{s_ops},
is a linear combination of convex polytopes $h_1, h_2, ..., h_k$ with weights $c_1,c_2,\cdots,c_k$.
\item $H(\VV,t)$ is a function defined in Section \ref{s_ops}.
\item $| X | = $ the size of a {\em multiset} or {\em set} $X$.
\item $\| a \| = $ the absolute value of a real number $a$.
\item $F$ denotes the {\em actual} set of faulty nodes in an execution of the algorithm.
\item $\phi=|F|$. Thus, $0\leq \phi\leq f$.
\item $F_v[t]$, $t\geq 0$, denotes the set of faulty nodes whose round $t$ execution is verified by at least one fault-free node, as per Definition \ref{d_verified}.
\item $\overline{F_v}[t]=F-F_v[t]$, $t\geq 0$.
\item $\alpha=1-\frac{1}{n}$.
\item 
We use boldface upper case letters to denote matrices, rows of matrices, and their elements. For instance, $\bfA$ denotes a matrix, $\bfA_i$ denotes the $i$-th row of matrix $\bfA$, and $\bfA_{ij}$ denotes the element at the intersection of the $i$-th row and the $j$-th column of matrix $\bfA$.

\end{itemize}

\section{$H_l(h_1, h_2,\cdots, h_\nu;~c_1, c_2,\cdots, c_\nu)$ is Convex}
\label{a:linear_convex}

\begin{claim}
$H_l(h_1, h_2,\cdots, h_\nu;~c_1, c_2,\cdots, c_\nu)$ defined in Definition \ref{def:linear_hull}
is convex.
\end{claim}
\begin{proof}

The proof is straightforward.

Let
\[ h_L := H_l(h_1, h_2,\cdots, h_\nu;~c_1, c_2,\cdots, c_\nu)\]
and
\[Q := \{ i ~|~ c_i\neq 0, ~ 1\leq i\leq \nu \}.\]
Given any two points $x, y$ in $h_L$, by Definition \ref{def:linear_hull}, we have
\begin{equation}
\label{eq:x2}
x = \sum_{i \in Q} c_i p_{(i, x)}~~\text{for some}~p_{(i,x)} \in h_i,~~i \in Q
\end{equation}

and 

\begin{equation}
\label{eq:y2}
y = \sum_{i \in Q} c_i p_{(i,y)}~~\text{for some}~p_{(i,y)} \in h_i,~~i\in Q
\end{equation}

Now, we show that any convex combination of $x$ and $y$ is also in $h_L$. Consider a point $z$ such that

\begin{equation}
\label{eq:z2}
z = \theta x + (1 - \theta) y~~~\text{where}~0 \leq \theta \leq 1
\end{equation}

Substituting (\ref{eq:x2}) and (\ref{eq:y2}) into (\ref{eq:z2}), we have

\begin{align}
z &= \theta \sum_{i\in Q}~c_i p_{(i,x)} + (1-\theta) \sum_{i \in Q}~c_i p_{(i,y)}\nonumber\\
&= \sum_{i\in Q}~c_i\left(\theta p_{(i,x)} + (1-\theta) p_{(i,y)}\right) \label{eq:p_iz2}
\end{align}

Define $p_{(i,z)} = \theta p_{(i,x)} + (1-\theta) p_{(i,y)}$ for all $i\in Q$. Since $h_i$ is convex, and $p_{(i,z)}$ is a convex combination of $p_{(i,x)}$ and $p_{(i,y)}$, $p_{(i,z)}$ is also in $h_i$. Substituting the definition of $p_{(i,z)}$ in (\ref{eq:p_iz2}), we have

\begin{align*}
z &=  \sum_{i \in Q}~~c_i~p_{(i,z)}~~\text{where}~p_{(i,z)} \in h_i,~~i\in Q
\end{align*}

Hence, by Definition \ref{def:linear_hull}, $z$ is also in $h_L$. Therefore, $h_L$ is convex.

\end{proof}


\section{Claim \ref{c_common_verify}}
\label{a_claim_common_verify}

\begin{claim}
\label{c_common_verify}
Consider fault-free nodes $i, j\in V-F$. Then

\begin{itemize}
\item 
If $(h,k,-1)\in \R_i[-1]$ at some point of time, then eventually, $(h,k,-1)\in\R_j[-1]$.

\item
For $t\geq 0$, if $(h,k,t-1)\in \R_i[t]$ at some point of time, then eventually $(h,k,t-1)\in \R_j[t]$.

\end{itemize}
\end{claim}

\begin{proof}

\noindent
{\bf First Part:}

In the preliminary round ($t = -1$),
node $i$ adds $(h,k)$ to $\R_i[-1]$ 
whenever it reliably receives message $(h,k,-1)$, i.e., $(h,k,-1)$ is either received by using stable vector or $\RBRecv(h,k,-1)$ occurred. (For messages in the preliminary round, $h$ is just
a single point.)
Then by {\em Global Liveness} property,
node $j$ will eventually reliably receive the same message, and add $(h,k,-1)$ to $\R_j[-1]$.

~

\noindent
{\bf Second Part:}

The proof is by induction.

{\em Induction basis:}
Suppose that in round $t = 0,$ at some real time
$\mu$, $(h,k,-1)\in \R_i[0]$.
Thus, node $i$ must have reliably received (at line 9 of round 0)
a message of the form $((h,\VV),k,-1)$ such that the following
conditions are true at time $\mu$:
\begin{itemize}
\item Condition 1: $\VV\subseteq \R_i[-1]$  (due to line 10, and the fact that $\R_i[-1]$ can
	only grow with time)
\item Condition 2: $|\VV| \geq n-f$ (due to Case $t=0$ in Procedure \Verify)
\end{itemize}
The {\em Global Liveness} property
implies that eventually node $j$ will also reliably receive the
message $((h,\VV),k,0)$ that was reliably received by node $i$.
Also, the correctness of the first part implies 
that eventually each element of $\R_i[-1]$ will be included
in $\R_j[-1]$. Thus, because $\VV\subseteq \R_i[-1]$ at time $\mu$, eventually $\VV\subseteq \R_j[-1]$.
As in Condition 2 above, node $j$ will also find that $|\VV| \geq n-f$.
Therefore, by lines 10-12, it follows that eventually
$(h,k,-1)\in \R_j[0]$.

{\em Induction:}
Consider round $t\geq 1$.
Assume that the second part of the lemma holds true through
rounds $t-1$. Therefore, if $(h,k, t-2)\in \R_i[t-1]$ at some point of time, then eventually
$(h,k, t-2)\in \R_j[t-1]$.

Now we will prove that the second part of the lemma holds for round $t$.
Suppose that at some time
$\mu$, $(h,k,t-1)\in \R_i[t]$.
Thus, node $i$ must have reliably received (at line 9 of round $t$)
a message of the form $((h,\VV),k,t)$ such that the following
conditions are true at time $\mu$:
\begin{itemize}
\item Condition 1: $\VV\subseteq \R_i[t-1]$  (due to 10, and the fact that $\R_i[t-1]$ can
	only grow with time)
\item Condition 2.1: when $t = 1$, $|\VV| \geq n-f$ and $h=H(\VV,0)$ (due to Case $t=1$ in Procedure \Verify)

\item Condition 2.2: when $t \geq 2$, $|\VV| \geq n-f$,
$h=H(\VV,t-1)$, and $(*,k,t-2)\in \VV$
 (due to Case $t\geq 2$ in Procedure \Verify)
\end{itemize}
The correctness of the second part of the lemma through round $t-1$ implies 
that eventually each element of $\R_i[t-1]$ will be included
in $\R_j[t-1]$. Thus, because $\VV\subseteq \R_i[t-1]$ at time $\mu$, eventually $\VV\subseteq \R_j[t-1]$.
Also, the {\em Global Liveness} property 
implies that eventually node $j$ will reliably receive the
message $((h,\VV),k,t)$ that was reliably received by node $i$; then, consider two cases:

\begin{itemize}
\item When $t = 1$, as in Condition 2.1 above, node $j$ will also find that
$|\VV| \geq n-f$, and $h=H(\VV,0)$.
Therefore, by lines 10-12, it follows that eventually
$(h,k,t-1)\in \R_j[t]$.

\item When $r \geq 2$, as in Condition 2.2 above, node $j$ will also find that
$|\VV| \geq n-f$, $(*,k,t-2)\in \VV$ and $h=H(\VV,t-1)$.
Therefore, by lines 10-12, it follows that eventually
$(h,k,t-1)\in \R_j[t]$.

\end{itemize}

Therefore, the proof for the second part is complete.

\end{proof}

\section{Proof of Lemma \ref{l_progress}}
\label{a_lemma_progress}

\noindent
{\bf Lemma \ref{l_progress}:}
{\em
Optimal Verified Averaging ensures {\em progress}: (i) all the fault-free nodes will eventually progress to round 0; and, (ii)
if all the fault-free nodes progress to the start of round $t$, $t \geq 0$, then all the
fault-free nodes will eventually progress to the start of round $t+1$.\\
}

\begin{proof}
\noindent
{\bf First Part:}

By assumption, all fault-free nodes begin the preliminary round eventually, and perform reliable broadcast
of their input (line 1).
Since the $(n-f)$ fault-free nodes follow the algorithm correctly,
$\SVRecv(-1)$ will eventually return (line 2). Therefore, node $i$ will eventually proceed to round 0 (line 5).


~
\noindent
{\bf Second Part:}

The proof is by induction.
By the first part, each fault-free node $i$ begins round 0 eventually, and performs reliable broadcast 
of $((h_i[-1],\R^c_i[-1]),i,0)$
on line 8.
Consider fault-free nodes $i,j$. 
By {\em Fault-Free Liveness} property of the primitives,
node $i$ will eventually reliable receive message $((h_j[-1],\R^c_j[-1]),j,0)$ from fault-free node $j$. By Claim \ref{c_common_verify}, eventually, $\R^c_j[-1] \subseteq \R_i[-1]$; therefore, node $i$ will progress past line 10. Moreover, since node $j$ is fault-free, it follows the algorithm specification correctly. Therefore, \Verify will return \TRUE, and node $i$ will eventually include $(\R^c_j[-1],j,-1)$. Since the above argument holds for all fault-free nodes $i,j$, it implies that
each fault-free node $i$ eventually adds $(\R^c_j[-1],j,-1)$ to $\R_i[0]$, for each fault-free node $j$ (including $j=i$). Therefore, at each fault-free node $i$,
eventually,
$|\R_i[0]|\geq n-f$, thus satisfying the condition
in Case $t=0$ of Procedure \Proceed.
Thus, Procedure \Proceed will return \TRUE, and each fault-free node $i$ will eventually proceed to round 1 (lines 13-16).

Now we assume that all the fault-free nodes have progressed to the start of round $t$,
where $t\geq 1$, and prove that all the fault-free nodes will eventually progress
to the start of round $t+1$.

Consider fault-free nodes $i,j\in V-F$.
At line 8 of round $t$, fault-free node $j$ performs reliable broadcast of $((h_j[t-1],\R^c_j[t-1]),j,t)$.
By {\em Fault-free Liveness} of reliable broadcast, fault-free node $i$ will eventually
reliably receive message $((h_j[t-1],\R^c_j[t-1]),j,t)$ 
from fault-free node $j$.
By Claim \ref{c_common_verify}, eventually 
$\R^c_j[t-1] \subseteq \R_i[t-1]$; therefore, node $i$ will progress past line 10
in the handler for message $((h_j[t-1],\R^c_j[t-1]),j,t)$.
Moreover, since node $j$ is fault-free, it follows the algorithm
specification correctly. Therefore, Procedure \Verify will return \TRUE in the handler
at node $i$ for message $((h_j[t-1],\R^c_j[t-1]),j,t)$ will all be correct.
Therefore, by lines 11-12, node $i$ will eventually include $(h_j[t-1],j,t-1)$ in $\R_i[t]$.
Since the above argument holds for all fault-free nodes $i,j$, it implies that
each fault-free node $i$ eventually adds $(h_j[t-1],j,t-1)$ to $\R_i[t]$, for each
fault-free node $j$ (including $j=i$). Therefore, at each fault-free node $i$,
eventually,
$|\R_i[t]|\geq n-f$,
and $(h_i[t-1],i,t-1)\in \R_i[t]$ (because the previous statement holds for $j=i$ too), thus satisfying both the conditions
in Case $t\geq 1$ of Procedure \Proceed.
Thus, Procedure \Proceed will return \TRUE, and each fault-free node $i$ will eventually proceed to round $t+1$ (lines 13-16).

\end{proof}

\section{Claims \ref{c_ver} and  \ref{claim:j_verified}}
\label{a_claims}

\begin{claim}
\label{c_ver}
If faulty node $i$'s round $t$ execution is verified by a fault-free node $j$, then
the following statements hold:


(i) For $t\geq 0$, $\R^c_i[t]\geq n-f$ and $h_i[t]=H(\R^c_i[t],t)$,

(ii) For $t\geq 0$, eventually $\R^c_i[t]\subseteq \R_j[t]$, and

(iii) For $t\geq 1$, node $i$'s round $t-1$ execution is also verified by node $j$.
\end{claim}
\begin{proof}
Let $t\geq 0$.
Suppose that node $i$'s round $t$ execution is verified by a fault-free node $j$.
In this case, we can use definitions (\ref{e_faulty_h}) and (\ref{e_faulty_V})
of $h_i[t]$ and $\R^c_i[t]$.
Definition \ref{d_verified} implies that node $j$ eventually 
reliably receives message $((h_i[t],\R^c_i[t]),i,t+1)$ from node $i$,
and subsequently adds (at line 12 in its round $t+1$) $(h_i[t],i,t)$ to $\R_j[t+1]$.
This implies that this message satisfies the checks done by node $j$ at lines 10 and 11:
Specifically,
(a) $|\R^c_i[t]|\geq n-f$,
(b) $h_i[t]=H(\R^c_i[t],t)$, and
(c) for $t \geq 1$, $(*,i,t-1)\in\R^c_i[t]$.
Also, by the time node $j$ adds $(h_i[t],i,t)$ to $\R_j[t+1]$, the condition
checked at line 10 also hold: specifically, $\R^c_i[t]\subseteq \R_j[t]$, proving
claim {\em (ii)} stated above.
Also, (a) and (b) above prove claim {\em (i)}.

For $t\geq 1$,
$(*,i,t-1)\in\R^c_i[t]$ and eventually $\R^c_i[t]\subseteq \R_j[t]$
together imply 
that eventually $(*,i,t-1)\in \R_j[t]$. Then
this observation together with 
Definition \ref{d_verified} imply that round $t-1$ execution of node $i$
is verified by node $j$. This proves claim {\em (iii)}.
\end{proof}

~

\begin{claim}
\label{claim:j_verified}
If faulty node $i$'s round $t$ execution is verified by a fault-free node $j$, $t \geq 0$,
then for all $r$ such that $0\leq r \leq t$, node $i$'s round $r$ execution is verified by node $j$.
\end{claim}

\begin{proof}
The claim is trivially true for $t=0$.
The proof of the claim for $t>0$ follows by repeated application of Claim \ref{c_ver}(iii) above.

\end{proof}

\section{Proof of Lemma \ref{lemma:J_in_H0}}
\label{a_lemma_J}

The proof of Lemma \ref{lemma:J_in_H0} uses the following theorem by Tverberg \cite{tverberg}:
\begin{theorem}
\label{thm:tverberg}
(Tverberg's Theorem \cite{tverberg}) For any integer $f \geq 0$, for every multiset $Y$ containing at least $(d+1)f+1$ points in a $d$-dimensional space, there exists a partition
$Y_1, .., Y_{f+1}$ of $Y$ into $f+1$ non-empty multisets such that $\cap_{l=1}^{f+1} \HH(Y_l) \neq \emptyset$.
\end{theorem}

~

\noindent
Now we prove Lemma \ref{lemma:J_in_H0}. \\

\noindent{\bf Lemma \ref{lemma:J_in_H0}:}
{\em
 For each node $i\in V-\overline{F_v}[0]$, the polytope $h_i[0]$ is non-empty. \\
}

\begin{proof}
%
Note that $V-\overline{F_v}[0]=(V-F)\cup F_v[0]$.
\begin{itemize}
\item
For a fault-free node $i\in V-F$, since it behaves correctly,
$|\R^c_i[0]|\geq n-f$ (due to the checks performed in \Verify), and
$h_i[0]=H(\R^c_i[0],0)$ (due to line 15).
\item
For faulty node $i\in F_v[0]$ as well, by Claim \ref{c_ver}(i)
in Appendix \ref{a_claims},
$|\R^c_i[0]|\geq n-f$ and
$h_i[0]=H(\R^c_i[0],0)$.
\end{itemize}
Thus, for each $i\in V-\overline{F_v}[0]$,
$|\R^c_i[0]|\geq n-f$ and
$h_i[0]=H(\R^c_i[0],0)$.

Consider any $i\in V-\overline{F_v}[0]$.
Consider the computation of polytope $h_i[0]$ as $H(\R^c_i[0],0)$.
By step 3 of Case $t=0$ in function $H$ in Section \ref{s_ops}, $|X|=|\R^c_i[0]|\geq n-f$.
Recall that, due to the lower bound on $n$ discussed in Section \ref{s_intro},
we assume $n \geq (d+2)f+1$.
Thus, in function $H$, $|X| \geq n-f \geq (d+1)f+1$. By Theorem \ref{thm:tverberg} above, there exists a partition $X_1, X_2, \cdots, X_{f+1}$ of $X$ into multisets $X_j$ such that $\cap_{j=1}^{f+1} \HH(X_j) \neq \emptyset$. Let us define
\begin{equation}
\label{eq:J}
J = \cap_{i=1}^{f+1} \HH(X_j)
\end{equation}
Thus, $J$ is non-empty.
In item (i.e., step) 4 of Case $t=0$ in function $H$, because $|X|\geq n-f$,
each multiset $C$ used in the computation of function $H$ is of size at least $n-2f$.
Thus, each $C$ excludes only $f$ elements of $X$, whereas there are $f+1$ multisets
in the above partition of $X$.
Therefore, each set $C$ in step 4 of item 1 of function $H$ will fully contain at least one multiset $X_j$
from the partition. Therefore, $\HH(C)$
will contain $J$. Since this holds true for all $C$'s, $J$ is contained in
the convex polytope computed by $H(\R^c_i[0],0)$.
Since $J$ is non-empty, $h_i[0]=H(\R^c_i[0],0)$ is non-empty.

\end{proof}

\section{Claim \ref{claim:notFv_not_verified}}
\label{app_s_claim:notFv_not_verified}

\begin{claim}
\label{claim:notFv_not_verified}
For $t\geq 0$,
if $b \in \overline{F_v}[t]$, then for all $i\in V-\overline{F_v}[t+1]$,
$(*,b,t)\not\in\R^c_i[t+1]$.
\end{claim}

\begin{proof}
Consider faulty node $b\in \overline{F_v}[t]$.
Note that $V-\overline{F_v}[t+1]=(V-F)\cup F_v[t+1]$.

\begin{itemize}
\item Consider a fault-free node $i\in V-F$.
Since $b\in \overline{F_v}[t]$, node $b$'s round $t$ execution is {\em not} verified
by {\em any} fault-free node. Therefore, by Definition \ref{d_verified},
for fault-free node $i\in V-F$, {\bf at all times},
$(*,b,t)\not\in \R_i[t+1]$.
Therefore, by line 14, $(*,b,t)\not\in \R^c_i[t+1]$.

\item
Consider a faulty node $i\in F_v[t+1]$.
In this case,
the proof is by contradiction. In particular, for some $h$, assume that $(h,b,t) \in \R^c_i[t+1]$.
Since $i\in F_v[t+1]$, there exists a fault-free node $j$ that verifies the
round $t+1$ execution of node $i$. Therefore, by Claim \ref{c_ver}(ii) in
Appendix \ref{a_claims},
eventually $\R^c_i[t+1]\subseteq \R_j[t+1]$.  This observation, along with the above assumption
that $(h,b,t) \in \R^c_i[t+1]$, implies that eventually $(h,b,t)\in \R_j[t+1]$.
Since node $j$ is fault-free, Definition \ref{d_verified} implies that execution of
node $b$ in round $t$ is verified, and hence $b\in F_v[t]$. This is a contradiction.
Therefore, $(*,b,t)\not\in \R^c_i[t+1]$.
\end{itemize}
\end{proof}

\section{Proof of Lemma \ref{lemma:always_notFv_not_verified}}
\label{a_always}

\noindent{\bf Lemma \ref{lemma:always_notFv_not_verified}:}
{\em
 For $r\geq 0$,
if $b \in \overline{F_v}[r]$, then for all $\tau\geq r$,
\begin{itemize}
\item $b \in \overline{F_v}[\tau]$, and 
\item for all $i\in V-\overline{F_v}[\tau+1]$, $(*,b,\tau)\not\in\R^c_i[\tau+1]$.
\end{itemize}
}

\begin{proof}
Recall that $F_v[r]\subseteq F$, and $\overline{F_v}[r] = F-F_v[r]$.

For $r \geq 0$, consider a faulty node $b \in \overline{F_v}[r]$.
Thus, $b\in F$.

We first prove that
$b \in \overline{F_v}[\tau]$, for $\tau\geq r$.
This is trivially true for $\tau=r$. So we only need to prove this
for $\tau>r$.
The proof is by contradiction.

Suppose that
there exists $\tau>r$ such that $b \not\in \overline{F_v}[\tau]$.
Thus, $b\in F_v[\tau]$.
The definition of $F_v[\tau]$ implies that 
node $b$'s round $\tau$ execution is verified by some fault-free node $j$.
Then Claim \ref{claim:j_verified} implies that
node $b$'s round $r$ execution is verified
by node $j$.
Hence by the definition of $F_v[r]$, $b\in F_v[r]$. This is a contradiction. 
This proves that $b\in\overline{F_v}[\tau]$.

Now, since $b\in\overline{F_v}[\tau]$, by
Claim \ref{claim:notFv_not_verified},
for all $i\in V-\overline{F_v}[\tau+1]$, $(*,b,\tau)\not\in\R^c_i[\tau+1]$.

\end{proof}

\section{Claims \ref{c_single}, \ref{c_common} and \ref{claim:at_least_one_verified}}
\label{a_claim:one}

\begin{claim}
\label{c_single}
For $t \geq -1$, a fault-free node $i$ adds at most one message from node $j$ to $\R_i[t]$, even if $j$ is faulty.
\end{claim}

\begin{proof}
As stated in the properties of the communication primitives in Section \ref{s_ops}, each fault-free
node $i$ will reliably receive at most one message of the form $(*,j,t)$ from node $j$ (either 
via \SVRecv~or via \RBRecv). Since $\R_i[t]$ only contains tuples corresponding to reliably received
messages, the claim follows.
\end{proof}

\begin{claim}
\label{c_common}
For $t\geq 1$, consider nodes $i,j\in V-\overline{F_v}[t]$.
If $(h,k,t)\in \R^c_i[t]$ and $(h',k,t)\in\R^c_j[t]$, then
$h=h'$.
\end{claim}
\begin{proof}
We consider four cases:
\begin{itemize}
\item $i,j\in V-F$:
In this case, due to {\em Global Uniqueness} property of the primitive,
nodes $i$ and $j$ cannot reliably receive different round $t$ messages from the same
node. Hence the claim follows.

\item $i\in V-F$ and $j\in F_v[t]$:
Suppose that fault-free node $p$ verifies round $t$ execution of node $j$. Then
by Claim \ref{c_ver}(ii), eventually $\R^c_j[t]\subseteq \R_p[t]$.
Since nodes $i$ and $p$ are both fault-free, similar to the previous
case, due to the {\em Global Uniqueness} property, nodes $i$ and $p$
cannot reliably receive distinct round $t$ messages.
Thus, if $(h,k,t)\in \R^c_i[t]$ and $(h',k,t)\in\R^c_j[t]\subseteq \R_p[t]$,
then $h=h'$.

\item $j\in V-F$ and $i\in F_v[t]$: This case is similar to the previous case.
\item $i,j\in F_v[t]$:
In this case, there exist fault-free nodes $k_i$ and $k_j$
that verify round $t$ execution of nodes $i$ and $j$, respectively.
Thus, by Claim \ref{c_ver}(ii), eventually $(h,i,t)\in \R^c_i[t]\subseteq \R_{k_i}[t]$ and $(h',i,t)\in\R^c_j[t]\subseteq \R_{k_j}[t]$.
Since $k_i,k_j$ are fault-free, {\em Global Uniqueness} implies that $h=h'$.
\end{itemize}
\end{proof}

\begin{claim}
\label{claim:at_least_one_verified}
For $t\geq 1$,
consider nodes $i, j \in V - \overline{F_v}[t]$. There exists a fault-free node $g\in V-F$ such that $(h_g[t-1],g,t-1) \in \R^c_i[t]\cap\R^c_j[t]$.
\end{claim}

\begin{proof}
For any fault-free node, say $p$, due to the conditions checked in Procedure \Proceed,
$|\R^c_p[t]|\geq n-f$.
For a node $k\in F_v[t]$, recall that $h_k[t]$ and $\R^c_k[t]$ are defined in
(\ref{e_faulty_h}) and (\ref{e_faulty_V}).
Thus, by Definition \ref{d_verified}, there exists some fault-free node, say $q$, that reliably receives
message $((h_k[t],\R^c_k[t]),k,t+1)$ from node $k$ in round $t+1$, and after performing checks
in Procedure \Verify, adds $(h_k[t],k,t)$ to $\R^c_q[t+1]$. The checks in Procedure \Verify,  performed by 
fault-free node $q$, ensure that $|\R^c_k[t]|\geq n-f$.

Above argument implies that for the nodes $i,j\in V-\overline{F_v}[t]$,
$\R^c_i[t]$ and $\R^c_j[t]$ both contain at least $n-f$ messages.
Therefore, by Claims \ref{c_single} and \ref{c_common}, there will be at least $n-2f \geq df+1\geq f+1$ elements in $\R^c_i[t]\cap\R^c_j[t]$.
Since $f$ is the upper bound on the number of faulty nodes,
at least one element in $\R^c_i[t]\cap\R^c_j[t]$ corresponds to
a fault-free node, say node $g\in V-F$. That is, there exists $g\in V-F$ such that
$(h_g[t-1],g,t-1) \in \R^c_i[t] \cap \R^c_j[t]$.

\end{proof}

\section{Proof of Lemma \ref{lemma:transition_matrix}}
\label{a_t_matrix}

\noindent{\bf Lemma 
\ref{lemma:transition_matrix}:}
{\em
 For $t\geq 1$,
transition matrix $\bfM[t]$ constructed using the above procedure satisfies the following conditions.
\begin{itemize}
\item For $i,j \in V$, there exists a fault-free node $g(i,j)$ such that $\bfM_{ig(i,j)}[t] \geq \frac{1}{n}$.

\item $\bfM[t]$ is a row stochastic matrix, and $\lambda(\bfM[t])\leq 1-\frac{1}{n}$.


\end{itemize}
}

\begin{proof}

\begin{itemize}
\item To prove the first claim in the lemma, we consider four cases for node pairs $i,j$.
\begin{list}{}{}
\item[(i)] $i,j\in V-\overline{F_v}[t]$:
 By Claim \ref{claim:at_least_one_verified},
there exists a node $g(i,j)$ such that
$(h_{g(i,j)}[t-1],g(i,j),t-1) \in \R^c_i[t]\cap \R^c_j[t]$.
By (\ref{eq:matrix_i}) in the procedure to construct $\bfM[t]$, $\matrixm_{ig(i,j)}[t]=\frac{1}{|\R^c_i[t]|}\geq \frac{1}{n}$
and 
$\matrixm_{jg(i,j)}[t]=\frac{1}{|\R^c_j[t]|}\geq\frac{1}{n}$.
\item[(ii)] $i\in\overline{F_v}[t]$ and $j\in V-\overline{F_v}[t]$:
$|\R^c_j[t]|\geq n-f$ elements of
$\matrixm_j[t]$ are equal to $\frac{1}{|\R^c_j[t]|}\geq \frac{1}{n}$.
Since $n-f\geq (d+1)f+1\geq 2f+1$, there exists a fault-free node $g(i,j)$ such that
$\matrixm_{jg(i,j)}\geq \frac{1}{n}$.
By (\ref{e_fv}), all elements of $\matrixm_i[t]$, including $\matrixm_{ig(i,j)}[t]=\frac{1}{n}$.
\item[(iii)] $j\in\overline{F_v}[t]$ and $i\in V-\overline{F_v}[t]$:
Similar to case (ii).
\item[(iv)] $i,j\in\overline{F_v}[t]$: 
By (\ref{e_fv}) in the procedure to construct $\bfM[t]$, all $n$ elements in $\matrixm_i[t]$ and $\matrixm_j[t]$
both equal $\frac{1}{n}$. Choose a fault-free node as node $g(i,j)$.
Then
$\matrixm_{ig(i,j)}[t]=
\matrixm_{ig(i,j)}[t]=\frac{1}{n}$.
\end{list}

\item Observe that, by construction, for each $i\in V$, the row vector $\bfM_i[t]$ 
is stochastic. Thus, $\bfM[t]$ is row stochastic.
Also, due to the claim proved in the previous item, and
Claim \ref{c_lambda_bound}, $\lambda(\matrixm[t])\leq 1-\frac{1}{n}<1$.

\end{itemize}
\end{proof}

\section{Proof of Lemma \ref{lemma:valid_initial_hull}}
\label{app_s_lemma:valid_initial_hull}

\noindent{\bf Lemma \ref{lemma:valid_initial_hull}:} {\em
 $h_i[0]$ for each node $i \in V - \overline{F_v}[0]$ is valid. \\
}


\begin{proof}
Recall that $V - \overline{F_v}[0] = (V-F) \cup F_v[0]$. Now, consider two cases:

\begin{itemize}
\item $i \in V - F$:
Recall that $h_i[0]$ is obtained using function $H(\R^c_i[0],0)$.
Note that the function $H(\R^c_i[0],0)$ first computes frequency counts $N(x,k)$ for each $(x,k)$, and then computes sets $Y$
and $X$ using $N(x,k)$ values.

For $X$ and $Y$ computed in $H(\R^c_i[0],0)$,
consider a value $x^*\in X$. Since $x^*\in X$, there must exist $k^*$
such that $(x^*,k^*)\in Y$.
This, in turn, implies that there must exist at least $f+1$ tuples of the
form $(\II,j,-1)\in \R^c_i[0]$ such that $(x^*,k^*,-1)\in\II$.
When $j$ is a fault-free node, $\II$ above must be equal to $\R_j^c[-1]$ due to the algorithm specification.
Since $(x^*,k^*,-1)$ appears in at least $f+1$ tuples as observed above, there exists at least
one fault-free node $j$ such that $(x^*,k^*,-1)\in\R_j^c[-1]$.
Therefore, if $k^*$ is fault-free, then $x^*$ must be the input vector at node $k^*$.

Also, by Claim \ref{c_single}, for any faulty node $b$, at most one tuple of the
form $(b,v)$ may appear in set $Y$ above. Therefore,
except for at most $f$ values in $X$ (which may correspond to faulty nodes), all the
other values in $X$ must be equal to inputs at fault-free nodes.
Therefore, at least one set $C$ used to compute {\tt temp} in step 4 in Case $t=0$ of function $H(\R^c_i[0],0)$
must contain only the inputs at fault-free nodes. 
Therefore, $h_i[0] = {\tt temp}$ is in the convex hull of the inputs at fault-free nodes.
That is, $h_i[0]$ is valid.

\item $i \in F_v[0]$:
Suppose that round $0$ execution of node $i$ is verified by a fault-free node $j$.
By Claim \ref{c_ver} in Appendix \ref{a_claims}, $h_i[0]=H(\R^c_i[0],0)$, $|\R^c_i[0]| \geq n-f$, and eventually
$\R^c_i[0]\subseteq \R_j[0]$.
Suppose that at some time $\tau$, $\R^c_i[0]\subseteq \R_j[0]$. Let $\R_j[0]$ at real time $\tau$ be
denoted as $\R^\tau_j[0]$. Then, $\R^c_i[0]\subseteq \R_j^\tau[0]$.
By an argument similar to the previous item, it should be easy to see that $H(\R_j^\tau[0],0)$ is {\em valid}.
Also, observe that if $\VV_1\subseteq \VV_2$, then $H(\VV_1,0)\subseteq H(\VV_2,0)$.
Thus, $H(\R^c_i[0],0)\subseteq H(\R_j^\tau[0],0)$, and since $H(\R_j^\tau[0],0)$ is valid,
$H(\R^c_i[0],0)$ is also valid.
Thus, $h_i[0]=H(\R^c_i[0],0)$ is valid.
\end{itemize}
\end{proof}

\section{Proof of Lemma \ref{lemma:linear_valid}}
\label{app_s_lemma:linear_valid}

The proof is straightforward, but included here for completeness.

\noindent{\bf Lemma \ref{lemma:linear_valid}:} {\em
 Suppose non-empty convex polytopes $h_1, h_2, \cdots, h_k$ are all valid. Consider $k$ constants $c_1, c_2, \cdots, c_k$ such that $0 \leq c_i \leq 1$ and $\sum_{i = 1}^k c_i = 1$.
Then the linear combination of these convex polytopes,
$H_l(h_1, h_2, \cdots, h_k; c_1, c_2, \cdots, c_k)$, is valid. \\
}

\begin{proof}

Observe that the points in $H_l(h_1,\cdots,h_k;c_1,\cdots,c_k)$ are convex
combinations of the points in $h_1,\cdots,h_k$, because $\sum_{i=1}^k c_i = 1$ and $0\leq c_i\leq 1$,
for $1\leq i\leq k$.
Let $G$ be the set of input vectors at the fault-free nodes
in $V-F$. Then, $\HH(G)$ is the convex hull of the inputs at the fault-free nodes.
Since $h_i$, $1\leq i\leq k$, is valid, each point $p\in h_i$ is in $\HH(G)$.
Since $\HH(G)$ is a convex polytope, it follows that any convex combination
of the points in $h_1,\cdots,h_k$ is also in $\HH(G)$.

\end{proof}

\section{Algebraic Manipulation in the Proof of Theorem \ref{thm:correctness}}
\label{a_hausdorff}

\begin{align}
d(p_i^*, p_j^*) &=  \sqrt{\sum_{l=1}^d (p_i^*(l) - p_j^*(l))^2} \nonumber\\
&= \sqrt{\sum_{l=1}^d \left(\sum_{k\in V-\overline{F_v}[0]} \bfM^*_{ik} p_k(l) - \sum_{k\in V-\overline{F_v}[0]} \bfM^*_{jk} p_k(l)\right)^2}\nonumber ~~~~ \text{by (\ref{pi}) and (\ref{pj}})\\
&= \sqrt{\sum_{l=1}^d \left(\sum_{k\in V-\overline{F_v}[0]} (\bfM^*_{ik}-\bfM^*_{jk}) p_k(l)\right)^2}\nonumber\\
&\leq \sqrt{\sum_{l=1}^d \left[\alpha^{2t}\left( \sum_{k\in V-\overline{F_v}[0]} \|p_k(l)\|\right)^2\right]} 
~~~~~~~ \mbox{~~~~ by (\ref{eq:delta1})}\nonumber\\
&= \alpha^t \sqrt{\sum_{l=1}^d \left(\sum_{k\in V-\overline{F_v}[0]} \|p_k(l)\|\right)^2}\label{eq:d_pipj}
\end{align}

\section{Proof of Lemma \ref{lemma:svSize}}
\label{a_l:svSize}


We first prove a claim that will be used in the proof of Lemma \ref{lemma:svSize}.
\begin{claim}
\label{claim:M*}
For $t \geq 1$, define $\bfM'[t] = \Pi_{\tau=1}^t \bfM[\tau]$. Then, for all nodes $j \in V - \overline{F_v}[t]$, and $k \in \overline{F_v}[0]$, $\bfM'_{jk}[t] = 0$.
\end{claim}

\begin{proof}
The proof is by induction on $t$.

{\em Induction Basis}: Consider the case when $t = 1$. Recall that $V - \overline{F_v}[1] = (V-F) \cup F_v[1]$. Consider
any $j\in V - \overline{F_v}[1]$, and $k \in \overline{F_v}[0]$. Then by 
Lemma \ref{lemma:always_notFv_not_verified}, $(*,k,0) \not\in \R^c_j[1]$.
Then, due to (\ref{eq:matrix_i-2}),
$\bfM_{jk}[1]=0$, and hence $\bfM'_{jk}[1]=\bfM_{jk}[1]=0$.


{\em Induction}: Consider $t \geq 2$. Assume that the claim holds true through  $t - 1$. 
Then, $\bfM'_{jk}[t-1]=0$ for all $j \in V - \overline{F_v}[t-1]$ and $k \in \overline{F_v}[0]$.
Recall that $\bfM'[t-1] = \Pi_{\tau=1}^{t-1} \bfM[\tau]$.

Now, we will prove that the claim holds true for $t$.
Consider $j \in V - \overline{F_v}[t]$
and $k\in \overline{F_v}[0]$.
Note that $\bfM'[t] = \Pi_{\tau=1}^{t} \bfM[\tau] = \bfM[t] \Pi_{\tau=1}^{t-1} \bfM[\tau] = \bfM[t] \bfM'[t-1]$.
Thus, $\bfM'_{jk}[t]$ can be non-zero only if there exists a $q \in V$ such that $\bfM_{jq}[t]$ and $\bfM'_{qk}[t-1]$
are both non-zero.

For any $q\in \overline{F_v}[t-1]$, by
Lemma (\ref{lemma:always_notFv_not_verified}), $(*,q,t-1) \not\in \R^c_j[t]$.
Then, due to (\ref{eq:matrix_i-2}), $\bfM_{jq}[t] = 0$ for all $q \in \overline{F_v}[t-1]$.
 Additionally, by the induction hypothesis,
 for all $q \in V - \overline{F_v}[t-1]$ and $k \in \overline{F_v}[0]$,
 $\bfM'_{qk}[t-1] = 0$.
Thus, these two observations together imply that there does not exist any $q \in V$ such that
$\bfM_{jq}[t]$ and $\bfM'_{qk}[t-1]$ are both non-zero.
Hence, $\bfM'_{jk}[t]=0$.
\end{proof}

~

~

\noindent{\bf Lemma \ref{lemma:svSize}:}
{\em
For all $i\in V-\overline{F_v}[t]$ and $t \geq 0$,
 $I_Z \subseteq h_i[t]$.
}

\begin{proof}
Recall that $Z$ and $I_Z$ are defined in (\ref{e_Z}) and (\ref{e_I_Z}), respectively.
We first prove that for all $i\in V-\overline{F_v}[0]$,
 $I_Z \subseteq h_i[0]$.

Recall that $V - \overline{F_v}[0] = (V-F) \cup F_v[0]$. Now, consider two cases:
\begin{itemize}
\item $i \in V-F$:

We first make the following observations for each fault-free node $i$:

\begin{itemize}
\item {\em Observation 1}: $\R^c_i[0]$ contains at least $f+1$ messages from fault-free nodes (at line 14). This is due to the {\em Fault-free Integrity} property of the primitive, and the fact that $|\R^c_i[0]| \geq n-f$ (due
to the condition checked in procedure \Proceed before $\R^c_i[0]$ is set equal to $\R_i[0]$).

\item {\em Observation 2}: $\R^c_i[0]$ contains tuples of the form $(\VV,*,-1)$. We will say that a tuple $(\VV,*,-1)\in
\R_i^c[0]$
contains $Z$ if $Z\subseteq \VV$.
Due to Observation 1, at least $f+1$ tuples in $\R^c_i[0]$ contain $Z$, because tuples
corresponding to all the fault-free nodes contain $Z$.

\item {\em Observation 3}: Observation 2 and the definition of $H$ imply that multiset $X$ defined in in step 3
of Case $t=0$ of function $H$ contains $X_Z$ defined in Section \ref{s_size}. 

\item {\em Observation 4}: Let $A$ and $B$ be sets of points in the $d$-dimensional space, where $|A|\geq n-f$, $|B| \geq n-f$ and $A \subseteq B$. Define $h_A := \cap_{\,C_A \subseteq A, |C_A| = |A| - f}~~\HH(C_A)$ and $h_B := \cap_{\,C_B \subseteq B, |C_B| = |B| - f}~~\HH(C_B)$. Then $h_A \subseteq h_B$. This observation follows directly from the fact that every multiset $C_A$ in the computation of $h_A$ is contained in some multiset $C_B$ used in the computation of $h_B$, and the property of $\HH$.
\end{itemize}

Now, consider the computation of $h_i[0]$ at line 13. By Observation 3 and Observation 4, $I_Z \subseteq H(\R^c_i[0],0) = h_i[0]$.

\item $i \in F_v[0]$:

Suppose that round 0 execution of node $i$ is verified by a fault-free node $j$.
By Claim \ref{c_ver}, eventually
$\R^c_i[0]\subseteq \R_j[0]$.
Since node $j$ is fault-free, $\R_j[0]$, and therefore, $\R^c_i[0]$, contains messages from at most $f$ faulty nodes. This together with the fact that $|\R^c_i[0]| \geq n-f$ (by Claim \ref{c_ver}), implies that $\R^c_i[0]$ contains messages from at least $f+1$ fault-free nodes. Then by this observation and the fact that $h_i[0] = H(\R^c_i[0],0)$ (by Claim \ref{c_ver}), we can show that $I_Z \subseteq h_i[0]$ using the same argument as in the previous case.

\end{itemize}
Thus, $I_Z\in h_i[0]$ for all $i\in V-\overline{F_v}[0]$. \\

\noindent
Now we make several observations for each fault-free node $i\in V-F$:
\begin{itemize}
\item As shown above, $I_Z\in h_j[0]$ for all $j\in V-\overline{F_v}[0]$. 
\item
From (\ref{e_Mstar}), for $t\geq 1$,
\[
\vectorv[t]=\bfM^* \vectorv[0]
\]
where $\vectorv_j[0]=h_j[0]$ for $j\in V-\overline{F_v}[0]$.
\item By Theorem \ref{t_M},  $\vectorv_i[t] = h_i[t]$.
\item
 Observe that $\bfM^*$ equals $\bfM'[t]$ defined in Claim \ref{claim:M*}.
Thus, due to Claim \ref{claim:M*}, $\bfM^*_{ik}=0$ for $k\in\overline{F_v}[0]$ (i.e,. $k\not\in V-\overline{F_v}[0]$).
\item

$\bfM^*$ is the product of row stochastic matrices; therefore, $\bfM^*$ itself is also row stochastic.
Thus, for fault-free node $i$, $\vectorv_i[t]=h_i[t]$ is obtained
as the product of the $i$-th row of $\bfM^*$, namely $\bfM^*_i$, and $\vectorv[0]$:
this product yields 
a linear combination of the elements of $\vectorv[0]$, where the weights
are non-negative and add to 1 (because $\bfM^*_i$ is a stochastic row vector).

\item
From (\ref{e_r_c}), recall that $\bfM_i^*\vectorv[0]=H_l(\vectorv[0]^T~;~\bfM^*_i)$.
Function $H_l$ ignores the input polytopes for which the corresponding weight is 0.
Finally, from the previous observations, we have that when the weight in $\bfM^*[i]$ 
is non-zero, the corresponding polytope in 
$\vectorv[0]^T$ contains $I_Z$. 
Therefore, the linear combination also contains $I_Z$.
\end{itemize}
Thus,
$I_Z$ is contained in $h_i[t]=\vectorv_i[t]=\bfM^*_i\vectorv[0]$.
\end{proof}

\section{Proof of Theorem \ref{thm:optSize}}
\label{a_t:optSize}

\noindent{\bf Theorem \ref{thm:optSize}:}
{\em
The output convex polytope at fault-free node $i$ using {\em Optimal Verified Averaging} is optimal as per Definition \ref{def:optSize}.
}

\begin{proof}
Consider set $X_Z$ defined in Section \ref{s_size}.
Due to Claim \ref{c_single} in Appendix \ref{a_claims} and the fact that set $X_Z$ contains at least $(n-f)$ tuples, at least $(n-2f)$ tuples in $X_Z$ correspond to inputs at fault-free nodes.
Let $V_Z$ denote the set of fault-free nodes whose tuples appears in $X_Z$.
Let $S=V-F-V_Z$. Since $|X_Z|\geq n-f$, $|S|\leq f$. 

Now consider the following execution of any algorithm ALGO that correctly solves Byzantine convex consensus. 
Suppose that the faulty nodes in $F$ follow the algorithm correctly except choosing an incorrect input (in acceptable range for inputs). Consider the case when nodes in $V-V_Z$, including fault-free nodes in $S$, are so slow that the other fault-free nodes must terminate before receiving any messages from the nodes in $V-V_Z$. The fault-free nodes in $V_Z$ cannot determine whether the nodes in $V-V_Z$ are just slow,
or faulty (crashed).

Nodes in $V_Z$ must be able to terminate without receiving any messages from the nodes in $V-V_Z$, including fault-free nodes in $S$. Thus, the output must be in the convex
hull of inputs at the fault-free nodes whose tuples are included in $X_Z$. However, any $f$ of the nodes whose values are
in $X_Z$ may be faulty. Therefore, the output obtained by ALGO must be contained
in $I_Z$ as defined in Section \ref{s_size}. 
On the other hand, by Lemma \ref{lemma:svSize}, the output obtained using {\em Optimal Verified Averaging} contains $I_Z$.
This proves the theorem.
\end{proof}

\end{document}